\documentclass[11pt]{article}

\usepackage{fullpage}
\usepackage[utf8]{inputenc}
\usepackage[T1]{fontenc}
\usepackage{amsmath}
\usepackage{amsthm}
\usepackage{enumitem}
\usepackage{txfonts}
\usepackage{enumerate}

\usepackage[usenames,dvipsnames]{xcolor}
\usepackage{tikz}
\usetikzlibrary{arrows,shapes,snakes,automata,backgrounds,petri,calc}

\tikzstyle{max}=[thick,draw,minimum size=1.4em,inner sep=0em]
\tikzstyle{min}=[diamond,thick,draw,minimum size=1.4em,
    inner sep=0em]
\tikzstyle{ran}=[circle,thick,draw,minimum size=1.4em,
    inner sep=0em]
\tikzstyle{act}=[circle,thick,draw,fill,minimum size=.7em,
    inner sep=0em]
\tikzstyle{mc}=[rounded corners,thick,draw,minimum size=1.4em,
    inner sep=.5ex]
\tikzstyle{tran}=[thick,draw,->,>=stealth]
\tikzstyle{loop left}=[tran, to path={.. controls +(150:.5)
    and +(210:.5) .. (\tikztotarget) \tikztonodes}]
\tikzstyle{loop right}=[tran, to path={.. controls +(30:.5)
    and +(330:.5) .. (\tikztotarget) \tikztonodes}]
\tikzstyle{loop above}=[tran, to path={.. controls +(60:.5)
    and +(120:.5) .. (\tikztotarget) \tikztonodes}]
\tikzstyle{loop below}=[tran, to path={.. controls +(240:.5)
    and +(300:.5) .. (\tikztotarget) \tikztonodes}]

\usepackage{subfigure}
\usepackage{url,xspace}

\usetikzlibrary{calc,shapes,arrows}
\usepackage{ifthen}
\usepackage{varioref}
\usepackage{xspace}
\usepackage{textcomp}
\usepackage{picins}

\newtheorem{theorem}{Theorem}[section]
\newtheorem{lemma}[theorem]{Lemma}

 \newtheorem{proposition}[theorem]{Proposition}
\newtheorem{remark}[theorem]{Remark}
\newtheorem{definition}[theorem]{Definition}

\newcommand{\Nset}{\mathbb{N}}

\newcommand{\Exp}{\mathbb{E}}

\newcommand{\calP}{\mathcal{P}}

\newcommand{\calO}{\mathcal{O}}

\newcommand{\dist}{\Delta}

\newcommand{\game}{\mathcal{G}}
\newcommand{\ch}{M}

\newcommand{\val}{\mathit{val}}
\newcommand{\histories}{\mathcal{H}}
\newcommand{\shistories}{\histories(\sigma)}
\newcommand{\wait}{\mathit{wait}}
\newcommand{\pay}{\mathit{payoff}}

\newcommand{\sig}{\mathit{S}}
\newcommand{\Node}{\mathit{Node}}
\newcommand{\offset}{\mathit{offset}}

\newcommand{\dmax}{\hat{d}}

\newcommand{\aval}{att\textrm{-}\val}

\newcommand{\contrref}[1]{\textbf{#1}}

\newcommand{\runs}{\mathcal{R}}
\newcommand{\Defended}{\mathcal{D}}

\newcommand{\su}{\hat{u}}

\newcommand{\success}{\mathit{succ}}
\newcommand{\suc}{\mathit{succ}}
\newcommand{\fail}{\mathit{fail}}
\newcommand{\last}{\mathit{last}}

\newcommand{\supp}{\mathit{supp}}

\newcommand{\size}[1]{||#1||}
\newcommand{\PTIME}{\textbf{P}}
\newcommand{\NPTIME}{\textbf{NP}}
\newcommand{\PSPACE}{\textbf{PSPACE}}
\newcommand{\Def}{\textsc{Defend}}
\newcommand{\Var}{\mathit{Var}}
\newcommand{\U}[1]{U\langle #1 \rangle}
\newcommand{\Circle}{\mathit{Circle}}

\newcommand{\Char}{\mathbf{Char}}

\newenvironment{claim}[1]{\par\noindent\underline{Claim #1:}\space}{}
\newenvironment{claimproof}[1]{\par\noindent\underline{Proof:}\space#1}{\hfill $\blacksquare$}

\newcommand{\theoremlike}[2]{\par\medskip\penalty-250
\refstepcounter{theorem}{{\bfseries\noindent#2 \ref{#1}.}}}

\newcommand{\thmhelperpre}[2]{\theoremlike{#1}{#2}}
\newcommand{\thmhelperpost}{\par\medskip}

\newenvironment{reftheorem}[1]{\thmhelperpre{#1}{Theorem}\it}{\thmhelperpost}

\setlist{itemsep=-.8ex,topsep=.5ex,leftmargin=3ex}

\title{Strategy Synthesis in Adversarial Patrolling Games}
\author{
Tom\'a\v{s} Br\'azdil \and
Petr Hlin\v{e}n\'{y} \and
Anton\'\i n Ku\v{c}era \and
Vojt\v{e}ch {\v{R}}eh\'ak \and
Mat\'{u}\v{s} Abaffy}
\date{Faculty of Informatics,
      Masaryk University,
      Botanick\'{a} 68a,
      Brno, Czech Republic\\
      \texttt{\{brazdil,hlineny,kucera,rehak\}@fi.muni.cz,bafco@mail.muni.cz}}

\begin{document}

\maketitle

\begin{abstract}
  Patrolling is one of the central problems in operational security.
  Formally, a \emph{patrolling problem} is specified by a set $U$ 
  of nodes (admissible defender's positions), a set $T \subseteq U$ of vulnerable \emph{targets},
  an environment $E \subseteq U \times U$ (admissible defender's moves), and a function $d$ which to every
  target $u$ assigns the time $d(u) \in \Nset$ needed to complete 
  an intrusion at~$u$. The goal is to design an optimal strategy
  for a \emph{defender} who is moving from node to node and
  aims at detecting possible intrusions at the targets. The defender can detect 
  an intrusion at a target $u$ only by visiting $u$ before the intrusion is
  completed. The goal of the \emph{attacker} is to maximize the
  probability of a successful attack, and the defender aims at the opposite. 
  We assume that the attacker
  is \emph{adversarial}, i.e., he knows the strategy of the defender
  and can observe her moves. 

  We prove that the defender has an optimal strategy for every 
  patrolling problem. Further, we show
  that for every $\varepsilon > 0$, there exists a finite-memory
  $\varepsilon$-optimal strategy for the defender constructible
  in exponential time (in the size of the game), and we observe that such a strategy cannot be computed
  in polynomial time unless $\PTIME = \NPTIME$.
  
  Since (sub)optimal strategy synthesis is computationally hard for
  patrolling problems in general environments, we continue our study
  by restricting ourselves to fully connected
  environments where $E = U \times U$ (where we can safely assume that $T = U$).
  Then, a patrolling problem is fully determined by its signature,
  i.e., a function $S$ such that $S(k)$ is the total number of targets with
  attack length equal to~$k$. We assume that $S$ is encoded by using binary numbers,
  i.e., the encoding size of $S$ can be exponentially smaller than
  the number of targets. We start by establishing an upper bound on the value of a given patrolling 
  problem, i.e., we bound the maximal probability of successfully defended attacks 
  that can be achieved by the defender against an arbitrary strategy 
  of the attacker. The bound is valid for an arbitrary patrolling problem such that $T =U$ and depends only on
  the signature~$S$.
  Then, we introduce a decomposition method which allows to split a given patrolling 
  problem $\game$ into smaller subproblems and construct a defender's strategy for $\game$ 
  by ``composing'' the strategies constructed for these subproblems. Using this method,
  we can synthesize (sub)optimal defender's strategies in time which is proportional
  to the encoding size of $S$. Consequently, we can compute (sub)optimal strategies
  for \emph{exponentially larger} patrolling problems then the existing methods based on 
  mathematical programming, where the size of the programs is proportional to the number
  of targets.
  Finally, for patrolling problems with $T = U$ and a well-formed signature, i.e., a signature $S$ such that
  $k$ divides $S(k)$ for every $k \in \Nset$, we give an exact 
  classification of all \emph{sufficiently connected} environments where the defender
  can achieve the same value as in the fully connected uniform environment. This result is useful
  for designing ``good'' environments where the defender can act optimally. 
\end{abstract}  

\pagebreak

\section{Introduction}
\label{sec-intro}

A central problem in security and operational research is how to deploy limited
security resources (such as police patrols, security guards, etc.) 
to maximize their effectiveness. Clearly, police patrols cannot be everywhere 
all the time, security guards cannot check every door every minute, etc., which
raises a crucial question how to utilize them best.
Game theoretic approaches to operational security problems based on Stackelberg model have received much attention
in recent years (see, e.g., \cite{tambe2011}). Informally, 
the  problem is to find the best possible strategy for a \emph{defender}
who is supervising potentially vulnerable targets (such as airports,
banks, or patrol stations) and aims at detecting possible 
\emph{intrusions}. The time needed to complete an intrusion at each 
target is finite, and the aim of the defender is to maximize the 
probability of discovering an intrusion before it is completed. 
An intensive research in this area has led to numerous successful applications (see, e.g., \cite{patrol-aplik-LAX,patrol-aplik-transit_system}). 
Due to high demand for practically usable solutions, the main emphasis has been put on inventing methods that can produce working solutions for large-scale instances quickly.  In most cases, the problem is 
simplified (for example, by restricting the set of defender's strategies
to some manageable subclass), and various tricks are used to avoid non-linear constraints and/or objectives. This approach enables efficient synthesis of strategies that are ``good enough'' for practical purposes (thus, the main engineering goal is achieved), but does not allow for
synthesizing optimal or \mbox{$\varepsilon$-strategies} (for a given $\varepsilon > 0$) in general. Further, the 
size of the resulting mathematical program is usually proportional 
to the number of targets, which influences the scalability of these 
methods. Since developing the basic theory of the underlying game model 
has not received so much attention as designing practically usable 
solutions, many fundamental questions (such as the computability of the Strackelberg value, the existence and computability of an \mbox{optimal/$\varepsilon$-optimal} defender's strategy, etc.) are open or have even been answered
incorrectly. In this paper, we provide a solution for some of these problems. As an unexpected payoff of our study, we also obtain a completely new approach to synthesizing defender's strategies in security games with fully connected environment based on compositional reasoning, which avoids 
the use of mathematical programming and can be applied to exponentially larger
instances than the currently available methods. A detailed explanation of the achieved results is given below.

In this paper, we consider the adversarial variant of patrolling, where the attacker is assumed to
be quite powerful---he can observe defender's moves, and he even knows defender's strategy. However, 
he cannot predict the way of resolving the defender's randomized choice. Formally,
a \emph{patrolling problem} $\game$ is specified by a finite set $U$ 
of nodes (possible defender's positions), a set $T \subseteq U$ of targets, an initial node $\su \in T$ (the initial 
position of the defender), an environment $E \subseteq U \times U$ (admissible moves of the defender) 
and a function $d : T \rightarrow \Nset$  which to every target associates the corresponding attack length. 
The defender starts at $\su$ and then moves from node to node consistently with $E$. We assume that 
traversing every edge takes precisely one unit of time (longer moves can be modeled by inserting
intermediate nodes.)
The defender may choose  the next node randomly and independently of 
her previous choices. Formally, a \emph{defender's strategy} is a function $\sigma : \histories \rightarrow \Delta(U)$ where $\histories$ is the set of all finite 
non-empty sequences of nodes and $\Delta(U)$ is the set of all 
probability distributions over~$U$. We require that $\sigma$ is consistent
with $E$, i.e., the support of $\sigma(h)$ is a subset of nodes that 
are immediate successors of the last node of~$h$. Note that each $\sigma$
determines a unique probability space over all \emph{runs} 
(infinite paths in $(U,E)$) initiated in $\su$ in the standard way, and we use  
$\calP^\sigma$ to denote the associated probability measure.

Depending on the observed walk of the defender, the attacker may choose to attack some
target or wait (we assume that the attacker may attack at most once
during a play). More precisely, an \emph{attacker's strategy} is function
$\pi : \histories \rightarrow T \cup \{\bot\}$ such that whenever $\pi(h)
\neq {\bot}$, then for all proper prefixes $h'$ of $h$ we have that
$\pi(h')={\bot}$. Since the attacker has a complete knowledge about
the current position of the defender, he would \emph{never} attack a
target currently visited by the defender. Still, he may attack this
target immediately after the defender's departure, i.e., long before
the defender arrives to the next node (think of an UAV patrolling
military bases). This assumption is reflected in the definition 
of a discovered attack---if the current location of the
defender is $u$ and the attacker attacks a target $v$, the defender
has to visit the node $v$ within the next $d(v)$ time units to discover
this attack, even if $u = v$. The aim of the defender is to maximize
the probability of successfully detected (or not initiated) attacks,
while the attacker aims at the opposite.  Given a strategy $\sigma$ of
the defender and a strategy $\pi$ of the attacker, we use
$\calP^{\sigma}(\Defended[\pi])$ to denote the probability of all
infinite paths $w$ initiated in $\su$ such that either $\pi(h) = {\bot}$ for
every prefix $h$ of $w$ (i.e., no attack is encountered along $w$), or
$\pi(h) = v \in T$ for some prefix $h$ of $w$ and $v$ is among the nodes visited
after $h$ in $w$ in at most $d(v)$ transitions (i.e., $w$ contains a successfully
defended attack). The \emph{value of $\sigma$} is defined by
$\val(\sigma) = \inf_{\pi} \, \calP^{\sigma}(\Defended[\pi])$, where
$\pi$ ranges over all strategies of the attacker.
The \emph{Stackelberg value} of $\game$ is defined by 
$\val \ = \ \sup_{\sigma} \, \val(\sigma)$, where $\sigma$ ranges over all strategies of the defender.
A defender's strategy $\sigma^*$ is \emph{$\varepsilon$-optimal} (where $\varepsilon \geq 0$) if 
$\val(\sigma^*) \geq \val - \varepsilon$. A \mbox{$0$-optimal} strategy is called \emph{optimal}.

\begin{remark}
   In our definition of the patrolling problem, we assume that all targets are equally important to
   the defender (and the attacker). The results \contrref{A} and \contrref{B} presented below remain valid even if we extend
   the model by assigning numerical \emph{weights} to nodes and modify the game objective so that
   the defender/attacker aims at maximizing/minimizing the expected weight of a discovered attack.
   If the weight (importance) of nodes is \emph{different} for each player, the game is no longer zero-sum,
   and the solution concept becomes somewhat different (consequently, our results do not apply in this case).
\end{remark}

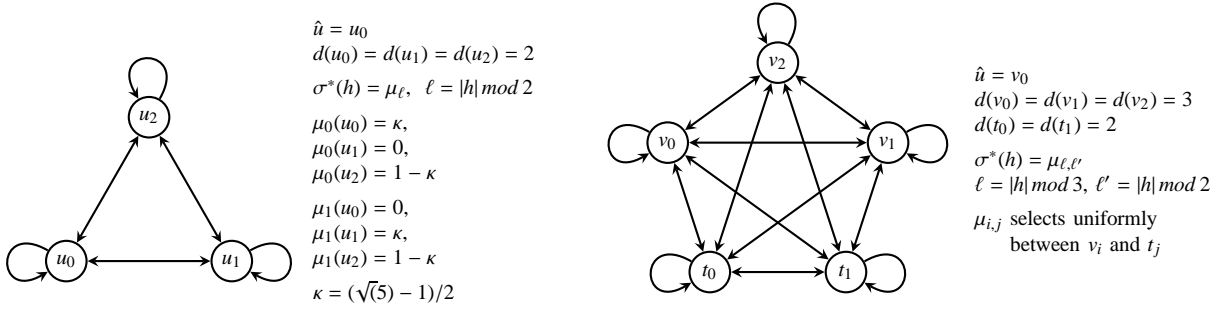
\begin{figure}
\begin{tikzpicture}[x=2.2cm,y=2.2cm,font=\scriptsize]
\node (v) at (0,-.13)    [ran] {$u_2$};
\node (u) at (-.5,-1) [ran] {$u_0$};
\node (s) at (.5,-1)  [ran] {$u_1$};
\draw [tran,<->] (v) -- (u);
\draw [tran,<->] (v) -- (s);
\draw [tran,<->] (u) -- (s);
\path (u) edge[loop left]  node[left=1pt] {} (u);
\path (s) edge[loop right] node[left=1pt] {} (s);
\path (v) edge[loop above] node[left=1pt] {} (v);
\node[text width=4cm,align=left] at (1.9,-.4) 
  {   $\su = u_0$\\
   $d(u_0) = d(u_1) = d(u_2) = 2$\\[1ex]
   $\sigma^*(h) = \mu_\ell$, \ $\ell = |h| \,\mathit{mod}\, 2$\\[1ex]
   $\mu_0(u_0) = \kappa$,\\
   $\mu_0(u_1) = 0$,\\
   $\mu_0(u_2) = 1- \kappa$\\[1ex]
   $\mu_1(u_0) = 0$,\\
   $\mu_1(u_1) = \kappa$,\\
   $\mu_1(u_2) = 1- \kappa$\\[1ex]
   $\kappa = (\!\sqrt(5) - 1)/2$
  };
\path (3.8,-.5) coordinate (origin);
\path (162-0:.7)    ++(origin) coordinate (P0);
\path (162-1*72:.7) ++(origin) coordinate (P1);
\path (162-2*72:.7) ++(origin) coordinate (P2);
\path (162-3*72:.7) ++(origin) coordinate (P3);
\path (162-4*72:.7) ++(origin) coordinate (P4);
\node (V0) at (P0)    [ran] {$v_0$};
\node (V1) at (P1)    [ran] {$v_2$};
\node (V2) at (P2)    [ran] {$v_1$};
\node (T0) at (P3)    [ran] {$t_1$};
\node (T1) at (P4)    [ran] {$t_0$};
\draw [tran,<->] (V0) -- (V1);
\draw [tran,<->] (V0) -- (V2);
\draw [tran,<->] (V0) -- (T0);
\draw [tran,<->] (V0) -- (T1);
\draw [tran,<->] (V1) -- (V2);
\draw [tran,<->] (V1) -- (T0);
\draw [tran,<->] (V1) -- (T1);
\draw [tran,<->] (V2) -- (T0);
\draw [tran,<->] (V2) -- (T1);
\draw [tran,<->] (T0) -- (T1);
\path (V0) edge[loop left]  node[left=1pt] {} (V0);
\path (V1) edge[loop above]  node[left=1pt] {} (V1);
\path (V2) edge[loop right]   node[left=1pt] {} (V2);
\path (T0) edge[loop right]   node[left=1pt] {} (T0);
\path (T1) edge[loop left]  node[left=1pt] {} (T1);
\node[text width=4cm,align=left] at (5.9,-.4) 
  {   $\su = v_0$\\
   $d(v_0) = d(v_1) = d(v_2) = 3$\\
   $d(t_0) = d(t_1) = 2$\\[1ex]
   $\sigma^*(h) = \mu_{\ell,\ell'}$\\
   $\ell = |h| \,\mathit{mod}\, 3$, $\ell' = |h| \,\mathit{mod}\, 2$\\[1ex]
   $\mu_{i,j}$ selects uniformly \\
   ~~~~ between $v_{i}$ and $t_{j}$
  };
\end{tikzpicture}
\caption{Two examples of patrolling problems and the corresponding 
optimal defender's strategies.}
\label{fig-example}
\end{figure}

\noindent
\textbf{Two simple examples.} To get some intuition about the patrolling 
problem, we start with two simple examples that will also be 
used to demonstrate some of our results. Let us first 
consider the patrolling problem of Fig.~\ref{fig-example}~(left).
Here, we need to patrol three nodes with the same attack length~$2$ (i.e., $T = U$), 
where $u_0$ is the initial node, in a fully connected environment. 
Let us try to determine the Stackelberg value and
an optimal strategy of the defender. A naive idea is to pick
a strategy $\sigma$ which always selects each of the
three immediate successors with probability $1/3$.
Consider a strategy
$\pi$ of the attacker such that $\pi(u_0) = u_2$. We have that
$\calP^{\sigma}(\Defended[\pi]) = 1/3 + 2/3 \cdot 1/3 = 5/9$, and one
can easily verify that for \emph{every} attacker's strategy $\pi'$
we have that $\calP^{\sigma}(\Defended[\pi']) \geq 5/9$. Hence,
$\val \geq \val(\sigma) = 5/9$. However, the defender can do 
better. Consider
the strategy $\sigma^*$ defined in Fig.~\ref{fig-example}~(left).
Observe that $\sigma^*$ is \emph{independent} of the currently
visited node; the only relevant information about the history
of a play is whether its \emph{length} is even or odd. If it is 
even (odd), then $\sigma^*$ randomly selects between $u_0$ and $u_2$ 
(or between $u_1$ and $u_2$) where the ratio between the two 
probabilities is the \emph{golden ratio}.
One can check that for every defender's strategy $\pi$ we have that
$\calP^{\sigma^*}(\Defended[\pi]) \geq (\!\sqrt{5}-1)/2$. Hence,
$\val \geq (\!\sqrt{5}-1)/2 > 5/9$. In fact, the 
strategy $\sigma^*$ is \emph{optimal}, i.e., $\val = (\!\sqrt{5}-1)/2$, which is perhaps 
unexpected (see also the paragraph ``\emph{Comments on \contrref{D}}'' below).

Now consider the patrolling problem of Fig.~\ref{fig-example}~(right). 
Here we need to patrol five nodes ($T=U$); two of them have the attack 
length~$2$ and three of them have the attack length~$3$.
Again, we assume a fully connected environment. If we examine a naive strategy $\sigma$ which
always selects the next node uniformly among all immediate successors, we
obtain that $\val(\sigma)= 9/25$. A better strategy $\sigma^*$ for the
defender is shown in Fig.~\ref{fig-example}~(right).  The strategy
$\sigma^*$ depends only on the length of the history modulo~$6$, and it
always chooses uniformly between exactly two nodes. It directly follows from
our subsequent contributions (namely \contrref{C}) that $\val =
\val(\sigma^*)= 1/2$, i.e., $\sigma^*$ is optimal and the Stackelberg value
is equal to $1/2$.

\smallskip 

\noindent
\textbf{Our contribution.} We start by proving the following results about the general
patrolling problem:   

\begin{itemize}
\item[\textbf{A.}] For an arbitrary patrolling problem, there exists an optimal strategy for the defender.
\item[\textbf{B.}] Given a patrolling problem $\game = (U,T,\hat{u},E,d)$ and
   a rational $\varepsilon > 0$, there is a \mbox{finite-memory} \mbox{$\varepsilon$-optimal}
   strategy $\sigma$ for the defender computable in time  exponential in $\size{\game}$ and polynomial in
   $\varepsilon^{-1}$ (here, $\size{\game}$ is the encoding size of $\game$, where
   the attack lengths are encoded in unary). Further, $\val(\sigma)$ is rational and
   can also be computed in exponential time, i.e., we can also approximate $\val$ up to a given $\varepsilon > 0$
   in exponential time. We also observe that $\val$ cannot be approximated up to the error smaller than
   $|U|^{-1}$ in polynomial time unless $\PTIME = \NPTIME$. 
\end{itemize}

\noindent
\emph{Comments on \contrref{A}.} The existence of optimal strategies for
patrolling problems (and their variants) has been claimed in previous works
(see, e.g., \cite{Basilico2009,Basilico2012}) by arguing in 
the following way. For each
$j \in \Nset$, let $\Sigma^j$ be the class of all defender's strategies
$\sigma$ such that $\sigma(h)$ depends only on the last~$j$ nodes of~$h$. 
If we restrict the range of $\sigma$ to the strategies of $\Sigma^j$ in
the definition of Stackelberg value, we obtain an
approximated value, denoted by $\val^j$.  
Obviously, $\val^{j+1} \geq \val^j$ for every $j \in \Nset$. 
By adapting the results of \cite{conitzer2006}, it has been shown
in \cite{Basilico2012} that for every $j \in \Nset$ one can compute a 
strategy $\sigma \in \Sigma^j$ which achieves the outcome $\val^j$
or better against every attacker's strategy. In 
\cite{Basilico2009,Basilico2012}, it has been also claimed that 
$\val = \val^j$ for some sufficiently large~$j$ (without providing any upper bound).
The argument is 
based on applying general results about strategic-form games, but 
a full proof is omitted. Using the techniques of Section~\ref{sec-well-formed},
we prove that this claim is \emph{incorrect}, even for the
simple patrolling problem of Fig.~\ref{fig-example}~(right) where 
the defender has \emph{no} optimal strategy in 
$\bigcup_{j=1}^\infty \Sigma^j$.  
In our proof of~\contrref{A}, we take
an infinite sequence of strategies $\sigma_1,\sigma_2,\ldots$
such that $\lim_{n \rightarrow \infty} \val(\sigma_n) = \val$
and ``extract'' and optimal strategy out of it.

\smallskip

\noindent
\emph{Comments on \contrref{B}.} Our exponential-time algorithm for constructing 
an \mbox{$\varepsilon$-optimal} strategy is based on combining two main ideas.
First, we show that the Stackelberg value of a given game stays the same
when the initial target is changed. This implies that small perturbations in
probability distributions employed by an optimal strategy cause only 
a small change in the strategy value. Hence, we can compute a suitable discretization
scale and safely restrict the range of considered strategies to the discretized 
probability distributions. Let $\dmax = \max_{u\in U} \{d(u)\}$. The next important observation
is that the $\dmax$-step behaviour of every strategy (after some finite history) can be fully 
characterized by a real-valued vector with exponentially many components, where each component corresponds to a probability of visiting some vertex in at most $k \leq \dmax$ transitions. Due to the previous discretization
step, we can safely restrict the range of these vectors to finitely (exponentially) 
many values. It follows that if there is \emph{some} \mbox{$\varepsilon$-optimal} strategy,
then there is also an \mbox{$\varepsilon$-optimal} strategy whose $\dmax$-step behaviour (after every
finite history) can be characterized by one of these exponentially many vectors, and we show how to
check the existence of such a strategy in exponential time (this is perhaps the most difficult part
of the argument).

The lower complexity bound is trivial. Given a patrolling problem with $d(u)=|U|=k$ for all $u\in U$, we have that $\val=1$ iff the environment contains a directed cycle
through all the nodes (i.e., it is a \emph{Hamiltonian digraph}),
which is \NPTIME-hard to decide.
If the game is a negative instance, then for every strategy of the defender, the attacker clearly can launch an
attack at the very beginning of a play with probability of success at least~$1/k$. From this we immediately obtain the second part of~\contrref{B}.  
Although in recent~\cite{DBLP:conf/fossacs/HoO15}, it is shown that the problem whether  $\val =1$
for a given patrolling problem is \PSPACE-complete,
the construction of~\cite{DBLP:conf/fossacs/HoO15} only (for principal reasons) rules out,
unless $\PTIME = \PSPACE$,
the existence of an $\varepsilon$-optimal strategy for the defender
with $\varepsilon\leq c\cdot\mathop{exp}\,(-|U|)$ for some~$c>0$.
\smallskip

Since solving general patrolling problems is computationally hard, we continue our study by restricting ourselves to \emph{fully connected} environments, where $E = U \times U$. Observe that the defender has no reason to visit non-target nodes in fully connected environments, and hence we can further safely assume that $T = U$. For example, think of a surveillance system equipped with several cameras installed in front of various doors, where the footage of the cameras is shown in turns on a single screen (for some small constant amount of time) watched by a human guard. The time needed to break (open and close) different doors can be different. Then, the nodes/targets of the associated 
patrolling problem correspond to the cameras, the environment is fully connected (assuming one can switch between the cameras freely), and the transition time between two nodes is the same (and it can be normalized to $1$). 
Under these assumptions, a patrolling problem is fully specified by its \emph{signature}, i.e., a function
$S : \Nset \rightarrow \Nset_0$ which for a given $k \in \Nset$ returns the number of all $u\in T$ 
with $d(u) = k$. An important subclass of signatures are \emph{well-formed} signatures, where $k$ divides
$S(k)$ for all $k \in \Nset$. For example, the signature of the patrolling problem of Fig.~\ref{fig-example}~(right) is well-formed, while the 
signature of the patrolling problem of Fig.~\ref{fig-example}~(left)
is not. We assume that signatures are represented using \emph{binary} 
numbers, i.e., the encoding size of $S$, denoted by $\size{S}$, can be \emph{exponentially smaller} than the number of nodes.

Before formulating our results about the patrolling problem in a fully connected environment, we need to explain one important \emph{conceptual} contribution of this paper, which is the notion of a \emph{modular} strategy and the associated \emph{compositionality} principle. A defender's strategy $\sigma$ is \emph{modular} if
$\sigma(h)$ depends only on the length of $h$ modulo some constant~$c$
(in particular, note that the current defender's position is irrelevant). For example, the two strategies of Fig.~\ref{fig-example} are modular (the constant $c$ is equal
to $2$ and $6$ for the strategy on the left and on the right, respectively). Let $\game$ be a patrolling problem with a set of nodes~$U$. For every
$U' \subseteq U$, let $\game[U']$ be the patrolling problem obtained
from $\game$ by restricting the set of nodes to $U'$ and the set of transitions to $E \cap U' {\times} U'$ (note that this
makes sense even if the environment of $\game$ is not fully connected). Let $U_1,\ldots, U_k \subseteq U$, and let
$\sigma_1,\ldots,\sigma_k$ be modular defender's strategies in
$\game[U_1],\ldots,\game[U_k]$, respectively. For every probability 
distribution $\nu$ over $\{1,\ldots,k\}$, we can construct the 
\emph{$\nu$-composition} of $\sigma_1,\ldots,\sigma_k$, which is a modular 
defender's strategy $\sigma$ in $\game[U_1\cup \cdots \cup U_k]$ 
defined by $\sigma(h) = \nu_1\cdot\sigma_1(h) + \cdots + \nu_k\cdot\sigma_k(h)$.
Note that $\sigma$ is a correctly defined defender's strategy for 
$\game[U_1\cup \cdots \cup U_k]$ only if the environment
of $\game$ contains all of the required transitions between the nodes
of $U_1,\ldots,U_k$ (if the environment of $\game$ is fully connected, 
this is no issue). It follows immediately that
$\val(\sigma) \geq \min \{\nu_i\cdot \val(\sigma_i) \mid 1 \leq i \leq k \}$ (as we shall see, this inequality can be \emph{strict}).
Thus, one can construct a defender's strategy for a given
patrolling problem $\game$ by splitting the set of nodes into two or
more subsets (not necessarily disjoint), solving the smaller instances recursively, and then 
computing a suitable convex combination of the solutions. As we shall see
momentarily, this approach leads to an efficient algorithm capable 
of computing optimal (or suboptimal) strategies for \emph{very} large patrolling problems in couple of seconds.

Now we can explain our main results about the patrolling problem in a fully connected environment. 
 
\begin{itemize}
\item[\textbf{C.}] Given a patrolling problem $\game$ where $T=U$, we have that \mbox{$\val \leq \left(\sum_{k \in \supp(S)} \frac{S(k)}{k}\right)^{-1}$}  where $S$ is the signature of $\game$ and $\supp(S)$ is the set of all $k \in \Nset$ such that
$S(k) > 0$. This bound is valid for an arbitrary environment~$E$.
\item[\textbf{D.}] There is an algorithm which inputs a signature $S$ of a patrolling problem $\game$ with a fully connected environment (where $T=U$) and outputs a pair $(\theta,V)$ such that the following conditions are satisfied:
  \begin{itemize}
  	 \item The running time of the algorithm in \emph{polynomial in $\size{S}$}. 
  	 \item $\theta$ is a symbolic representation of a modular strategy for $\game$, and $V$ is a symbolic representation
  	    of $\val(\theta)$. Both $\theta$ and $V$ are parameterized by variables $\{p_1,\ldots,p_k\}$, where $k$ is bounded by a polynomial in $\size{S}$. The values of  $\{p_1,\ldots,p_k\}$ correspond to the unique solution (in $[0,1]^k$) 
  	    of a recursive system of polynomial equations that is also constructed by the algorithm. The number of variables $k$
  	    actually depends on the ``Euclid complexity'' of $S$ and can be constant (or even zero) for arbitrarily large~$S$. 
  	 \item If the signature $S$ is well-formed, then $k=0$  and the strategy $\theta$ is optimal.
  	    Since $k=0$, no extra computational time is needed to calculate/approximate the parameters, and hence $\theta$ is ``fully synthesized'' in time polynomial in $\size{S}$. 
  	 \item If the signature $S$ is not well-formed, then the
  	    strategy $\theta$ is a $\nu$-composition of simpler modular strategies and the variables defined via the
  	    system of polynomial equations correspond to the weights used to combine these simpler strategies together. 
  	    Further, we have that $\val_d < \val(\theta) < \val_u$, where
        $\val_d$ and $\val_u$ are the Stackelberg values of the patrolling
        problems with signatures $S\!_d$ and $S\!_u$ defined by
        $S\!_d(k) = k \cdot \lfloor \frac{n}{k} \rfloor$ and
        $S\!_u(k) = k \cdot \lceil \frac{n}{k} \rceil$, respectively.	    
  \end{itemize}
\item[\textbf{E.}] Given a patrolling problem $\game$ 
  with $T = U$ and a well-formed
  attack signature $S$, we say that the environment 
  $E$  of $\game$ is \emph{sufficiently connected} if  $\val$ is equal to the value 
  of $\game$ in the \emph{fully connected} environment. The problem whether $E$ is sufficiently connected  is 
  \mbox{\textbf{NP}-complete}. Further,
  this problem is \mbox{\textbf{NP}-complete} even for
  a subclass of patrolling problems such that $\supp(S) = \{k\}$, where 
  $k \geq 3$ is a fixed constant. For a subclass of patrolling
  problems where $\supp(S) = \{2\}$, the problem is solvable in 
  polynomial time.
                \end{itemize}

\noindent
\emph{Comments on \contrref{C}.} Note that the presented upper bound on
$\val$ does not depend on~$E$. An obvious question is whether this bound is
\emph{tight}. That is, given a function $S : \Nset \rightarrow \Nset_0$ such
that $\supp(S)$ is finite, we ask whether there exists a patrolling problem
$\game$ with $T = U$ such that the signature 
of $\game$ is $S$ and
\mbox{$\val = 1/(\sum_{k \in \supp(S)} S(k)/k)$}.  It follows from our results that
the answer to this question is \emph{yes} if $S$ is well formed. This means that the bound can be potentially lowered
(only) for those $S$ that are not well formed.  As an example, consider the
patrolling problem of Fig.~\ref{fig-example}~(left). Here $\supp(S) = \{2\}$
and $S(2) = 3$, and hence we obtain $\val \leq 2/3$. Since $\val =
(\!\sqrt{5}-1)/2 < 2/3$, the bound is not tight. For the patrolling problem
of Fig.~\ref{fig-example}~(right) we have that $\supp(S) = \{2,3\}$, $S(2) =
2$, and $S(3) = 3$, which gives an upper bound $(2/2 + 3/3)^{-1} =
1/2$. Since $\val = 1/2$, this bound is tight.  \smallskip

\noindent
\emph{Comments on \contrref{D}.}
The strategy $\theta$ is obtained by applying the ``decomposition''
technique described earlier. Since we intend to produce a strategy synthesis algorithm whose
running time is polynomial in $\size{S}$, we also need to design a special 
language allowing for compact representation of modular strategies in space polynomial
in $\size{S}$ (see Section~\ref{sec-modular-strategies}).  First, we split the nodes of $\game$ into
disjoint subsets according to their attack length. Then, we show how to compute a modular strategy for a set of $n$ nodes with the same attack 
length $d$. Here, we use a decomposition technique which resembles Euclid's gcd algorithm. First we check whether $d$ divides $n$. If so, we split the $n$ nodes into pairwise disjoint sets $U_0,\ldots,U_{d-1}$ so that $|U_i| = n/d$ for every
$0\leq i < d$, and define a modular strategy $\sigma$ such that 
$\sigma(h)$ selects uniformly among the elements of $U_i$, where
$i = |h| \mathrm{~mod~} d$. Observe that 
$\val(\sigma) = d/n$, which is optimal by \contrref{C}. If $d$ does not
divide $n$ and $n = k\cdot d + c$ where $1 \leq c < d$, then we split
the $n$ nodes into two disjoint subsets $U_1$ and $U_2$, where
$U_1$ contains $k \cdot d$ nodes and $U_2$ contains $c$ nodes.
A strategy $\sigma_1$ for $U_1$ is constructed as above, and we need
to process the set $U_2$. If $c$ divides $d$, the strategy $\sigma_2$
for $U_2$ is a simple loop over the nodes of $U_2$.
A closer look reveals that an appropriate distribution 
$\nu = (\nu_1,\nu_2)$ for combining $\sigma_1$ and $\sigma_2$ should 
satisfy the equation 
$\nu_1 \cdot \val(\sigma_1) = 1 - \nu_1^{d/c}$ which says that the nodes of
$U_1$ and $U_2$ are defended equally well. If  
$c$ does not divide $d$ and $d = j \cdot c + t$, where $1 \leq t < c$,
then the strategy $\sigma_2$ for $U_2$ spends the first $j\cdot c$ steps by 
performing the simple loop over the nodes of $U_2$, and the next $t$ steps by behaving exactly as
the strategy constructed for $|U_2|$ nodes with attack length $t$ (which is constructed
recursively). Then, $\sigma_2$ just keeps repeating its first $d$~steps. Again, we can setup an equation that should be satisfied by an appropriated distribution which combines $\sigma_1$ and $\sigma_2$ so that all targets are protected
equally well. This procedure eventually produces a modular strategy for defending $n$ nodes with
the same attack length~$d$. If $d$ divides $n$, then this strategy is provably optimal. In fact, we conjecture
that the constructed strategy is \emph{always} optimal, but we leave this hypothesis open (recently,
it has been shown by Lamser \cite{Lamser:BCthesis} that the algorithm produces an optimal strategy for all odd $n$ and 
$d=2$). Further, let us note that the number of variables/equations in the constructed system of polynomial equations 
is bounded by a polynomial in $\size{S}$, but the size of $S$ 
is \emph{not} a good measure for identifying hard instances. 
What really matters is the number of ``swaps'' in the Euclid's algorithm applied to $n$ and $d$; see  
Section~\ref{sec-modular-strategies} for further comments.  
After processing all subsets of nodes with the same attack length, we
combine the resulting strategies using an appropriate distribution.
The details are given in Section~\ref{sec-modular-strategies}. 

As an example, consider the patrolling problems of Fig.~\ref{fig-example}.
In the first case, we have $3$~nodes with the same attack length~$2$.
Since $2$ does not divide $3$, we split the set of nodes into
$U_1 = \{u_0,u_1\}$ and $U_2=\{u_2\}$. The strategy $\sigma_1$ for $U_1$
selects the node $u_1$ or $u_0$ with probability $1$, depending on whether
the length of the history is odd or even, respectively. Note that
$\val(\sigma_1)  = 1$. For the set $U_2$, we have that $|U_2|$ divides
$2$, and so the strategy $\sigma_2$ is a self-loop on $u_2$. The appropriate
distribution  $\nu = (\nu_1,\nu_2)$ for combining $\sigma_1$ and $\sigma_2$ should satisfy the equation $\nu_1 = 1 - \nu_1^{2}$. Thus, we obtain that
$\nu = \kappa = (\!\sqrt{5} -1)/2$, which yields the strategy of Fig.~\ref{fig-example}~(left). The strategy of Fig.~\ref{fig-example}~(right)
is obtained by first splitting the set of nodes into $U_1 = \{t_0,t_1\}$ and
$U_2 = \{v_0,v_1,v_2\}$ according to their attack length, solving these
subproblems (note that the solution for $U_i$ is a strategy which loops 
over the vertices of $U_i$), and then combining them with $\nu = (0.5, 0.5)$. 
\smallskip

\noindent
\emph{Comments on \contrref{E}.} We show that for every patrolling 
problem $\game = (U,T,\su,E,d)$ with $T=U$ and a well formed signature~$S$, there
exists a \emph{characteristic digraph} $\ch_S$ depending only on $S$
and computable in polynomial time, such that $E$ is sufficiently connected \emph{if, and only if,}
$(U,E)$ contains a subdigraph isomorphic (respecting the attack lengths) to $\ch_S$.  From this we immediately obtain that the problem whether a given $E$ is sufficiently connected
is in~\textbf{NP}, and we also provide the matching lower bound.
		Note that the characteristic digraph can be used to \emph{synthesize}
a minimal sufficiently connected environment for solving a given 
patrolling problem.

\smallskip

\noindent
\textbf{Related work.} 
Two player zero-sum stochastic games with both perfect and imperfect
information have been studied very intensively in recent years
(see, e.g., \cite{ChatterjeeH12,HansenMZ13,HansenKLMT11}), 
also for games with infinite state-space
\cite{BBKO:BPA-games-reachability-IC,EY:RSCG,EtessamiWY08,AbdullaCMS13}.
Patrolling games have so far been considered mainly in the context of
operation research. Here, the emphasis is usually put on finding methods
allowing to synthesize a sufficiently good defender's strategy, and 
the basic theoretical questions related to the underlying formal model 
are usually not studied in greater detail. The problem of finding
locally optimal strategies for robotic patrolling units have been studied
either in restricted environments (e.g., on circles in
\cite{AgmonKK08,AgmonKK08-2}), or fully-connected environments with weighted
preference on the targets \cite{Basilico2009,Basilico2009-2}.
Some novel aspects of the problem, such as variants with moving 
targets \cite{bosansky2011aamas,Fang2013}, multiple patrolling 
units \cite{Basilico2010}, or movement 
of the attacker on the graph \cite{Basilico2009-2} and reaction 
to alarms \cite{MunozdeCote2013} have also been considered in recent
works.

\section{The results}
\label{sec-results}

\newcommand{\croot}{\mathbf{r}}
\newcommand{\cfirst}{\mathbf{s}}
\newcommand{\commit}{\mathbf{c}}

\newcommand{\CharA}[3]{\mathbf{H}[#1,#2]_{#3}}
\newcommand{\CharB}[2]{\mathbf{H}[#1,#2]}
\newcommand{\CharC}[1]{\mathbf{H}[#1]}

We assume familiarity with the notions introduced earlier in Section~\ref{sec-intro}.

\subsection{The existence of an optimal defender's strategy}

We start by proving that there exists an optimal strategy for 
the defender. This is a generalization of similar results 
recently achieved in \cite{ABRBK} for a special
type of patrolling games where all nodes share the same attack length
(i.e., $\supp(S)$ is a singleton). The proof technique is completely different.

\begin{theorem}
\label{thm:optimal}
  For every patrolling problem $\game = (U,T,\su,E,d)$,  there exists an optimal defender's  strategy.
\end{theorem}
\begin{proof}[Proof Sketch]
We construct an optimal strategy $\sigma^*$ as a point-wise limit of a sequence $\sigma^1,\sigma^2,\ldots$ of strategies where each $\sigma^k$ is $1/k$-optimal. More precisely, we select $\sigma^1,\sigma^2,\ldots$ in such a way that for each history $h$, the sequence of distributions  $\sigma^1(h),\sigma^2(h),\ldots$ converges to a probability distribution, and we define $\sigma^*(h)$ to be its limit (we obtain $\sigma^1,\sigma^2,\ldots$ by starting with an arbitrary sequence of $1/k$-optimal strategies and successively filtering subsequences that are convergent on individual histories). It is relatively straightforward to show that if $\val(\sigma^*)\leq\val-\delta$ for some $\delta>0$, then for all $k$'s large enough we have $\val(\sigma^k)\leq \val-\delta/2$, which contradicts the fact that each $\sigma^k$ is $1/k$-optimal. For details see Appendix~\ref{app-optimal}.
\end{proof}

\noindent

\subsection{Computing finite-memory $\varepsilon$-optimal strategies}

In this subsection we describe a generic algorithm which for a given patrolling problem computes a finite representation of an $\varepsilon$-optimal strategy. Let us start with the definition of a finite-memory strategy.

\begin{definition}
A {\em finite-memory defender's strategy} is a tuple $(M,N,m_0,\xi)$ where $M$ is a finite set of memory elements, $N:M\times U\rightarrow M$ assigns to every memory element $m\in M$ and a node $u\in U$ a next memory element $N(m,u)$, $m_0$ is an initial memory element, and $\xi:M\times U\rightarrow \dist(U)$ is a function which to every memory element $m\in M$ and a node $u\in U$ assigns a distribution $\xi(m,u)$ on $U$ such that $\supp(\xi(m,u))\subseteq \suc(u)$.

A finite-memory defender's strategy $(M,N,m_0,\xi)$ induces a defender's
strategy $\sigma$ as follows: We extend $N$ to an "empty" history $\varepsilon$ by $N(m_0,\varepsilon)=m_0$, and to all histories $hv\in \histories$,
here $v\in U$, inductively by $N(m_0,hv)=N(N(m_0,h),v)$. Then for $hu\in \histories$ (where $u\in U$) we have that $\sigma(hu)=\xi(N(m_0,h),u)$. 
\end{definition}
\noindent
\begin{theorem}\label{thm:eps-opt}
Let $\varepsilon>0$ and assume that $\su\in T$. There is an $\varepsilon$-optimal finite-memory defender's strategy computable in time 
\[
\left(\frac{\dmax\cdot |U|}{\varepsilon}\right)^{\calO(\dmax^2\cdot |U|^2)}.
\]
\end{theorem}
\noindent
We construct our strategy using the so-called \emph{characteristics} (some intuition is given below). 
\begin{definition}
  A {\em characteristic} $c$ is a triple $(\croot,\cfirst,\commit)$ where $\croot\in U$, $\cfirst$ is a probability distribution on $U$, and $\commit:\{2,\ldots,\dmax\}\times
  T\rightarrow [0,1]$. Denote by $\Char$ the
  set of all characteristics.  Given $c=(\croot,\cfirst,\commit)\in \Char$, we denote by $\val(c)$
  the value $\min_{u\in T} \commit(d(u),u)$ of $c$.
\end{definition}
\noindent
Given a characteristic $c$, we use $c_{\croot}, c_{\cfirst}, c_{\commit}$ to denote the three components of $c=(\croot,\cfirst,\commit)$, respectively.

Intuitively, we interpret a given characteristic $c$ as a "local" plan of defence for next $\dmax$ steps where
\begin{itemize}
\item $c_{\croot}$ is the current node, 
\item $c_{\cfirst}$ is the current assignment of probabilities to the successors
  of $c_{\croot}$, and
\item for every $2\leq k\leq \dmax$ and every $u\in T$, we interpret $c_{\commit}(k,u)$ as the probability of visiting $u$ in at least one, and at most $k$ steps from
  $c_{\croot}$.~\footnote{Note that many characteristics are not
   ``consistent'' (if e.g. $c_{\cfirst}(u)=1/2$ and $c_{\commit}(1,u)=1/4$). But later we make
    sure that only consistent characteristics are used.}
\end{itemize}
To simplify our notation, we denote by $c_{\commit}(1,u)$ the probability $c_{\cfirst}(u)$ for every $u\in T$.

Now assume that the current plan is formalized by a characteristic $c$, and suppose that the defender makes one step to a next vertex $v$ chosen randomly with probability $c_{\cfirst}(v)$. Now the defender declares a new plan, $c^v\in \Char$ where $c^v_{\croot}=v$. However, the crucial observation is that the new plans $(c^v)_{v\in U}$ must be consistent with the original plan $c$ in the following sense for all $2\leq k\leq \dmax$ and all $u\in T$ :
\[
c_{\commit}(k,u) =	c_{\cfirst}(u)+\sum_{v\not = u} c_{\cfirst}(v)\cdot c^v_{\commit}(k-1,u) 
\]
We say that such a~vector $(c^v)_{v\in U}\in \Char^U$ of characteristics is a {\em successor} of $c$.

Now let $C$ be a finite set of characteristics such that every $c\in C$ has
a successor $(c^v)_{v\in U}\in C^U$ (i.e., $c^v\in C$ for all $v\in U$), and
there is at least one $\hat{c}\in C$ such that $\hat{c}_{\croot}=\su$. We say that such $C$ is {\em closed}. We construct a~finite-memory strategy $(M,N,m_0,\xi)$ where $M=C$, $N(c,v)=c^v$, $m_0=\hat{c}$, and $\xi(c)=c_{\cfirst}$. Intuitively, the strategy follows the plans in $C$ and always proceeds to the next plan according to a fixed successor in $C^U$. We prove that this strategy works consistently with the characteristics of $C$, i.e., whenever the current history is $h$ and the current memory element is $c$, then, subsequently, the probability of reaching $u$ in at least one, and at most $k$ steps is equal to $c_{\cfirst}(k,u)$. Thus the value of the finite-memory strategy cannot be worse than $\min_{c\in C} \val(c)$.

So, the computation of a finite-memory strategy reduces to a computation of a
finite closed set of characteristics. We show that one such set can be
extracted from a carefully selected $\varepsilon$-optimal strategy. Given a
defender's strategy $\sigma$, we denote by $\histories(\sigma)$ the set of
all histories that $\sigma$ may follow with a positive probability. Given a
strategy $\sigma$ and a history $h\in \histories(\sigma)$, we define a
characteristic $c[\sigma,h]$ such that $c[h]_{\croot}$ is the last node of
$h$, $c[h]_{\cfirst}=\sigma(h)$, and each $c[h]_{\commit}(k,u)$ is the
probability of reaching $u$ in at least one, and at most $k$ steps starting
with the history $h$ using $\sigma$. Now let $\sigma^*$ be an optimal
strategy. The crucial observation  (see also Proposition~\ref{cor:eps-opt-discr} in
Appendix~\ref{app-eps-opt}) is that for every $h\in \histories(\sigma)$ it
holds that $\val(c[\sigma^*,h])\geq \val$. By appropriately rounding probabilities in $\sigma^*$, we obtain an $\varepsilon$-optimal strategy $\sigma^{\varepsilon}$ such that for every history $h$ and every $u\in U$ :
\[
\sigma_{\varepsilon}(h)(u)=k\cdot \lceil \dmax\cdot|U|/\varepsilon\rceil^{-1}\text{ for a suitable }k\in \Nset 
\]
and $c[\sigma^{\varepsilon},h]\geq \val-\varepsilon$ for all $h\in \histories(\sigma^{\varepsilon})$.

Now it is rather straightforward to show that for each $h$, the vector $(c[hv])_{v\in U}$ is a successor of $c[h]$. Thus the set $\Char[\sigma^{\varepsilon}]$ of all $c[h]$, here $h\in \histories$, is a closed set. It is also finite, of size that is bounded by $\left(\dmax\cdot |U|/\varepsilon\right)^{\calO(\dmax^2\cdot |U|)}$, and every $c\in \Char[\sigma^{\varepsilon}]$ satisfies $\val(c)\geq \val-\varepsilon$. This shows that there always exists a $\varepsilon$-optimal finite-memory strategy of the size bounded by $\left(\dmax\cdot |U| /\varepsilon\right)^{\calO(\dmax^2\cdot |U|)}$. 

Our algorithm computes a closed subset $C$ of a (finite) set of appropriately rounded characteristics that maximizes $\min_{c\in C}\val(c)$. This is done by a simple iterative procedure which maintains a growing pool of characteristics (in order of decreasing value) and tries to find its closed subset. For details see Appendix~\ref{app-eps-opt}.

\subsection{A bound on the Stackelberg value}

Now we establish an upper bound on $\val$ which depends only
on the attack signature $\sig$ of $\game$. The simplicity of the argument 
is due to Proposition~\ref{prop-abafy}.       
\begin{theorem}
\label{thm-upper}
  For every patrolling problem $\game = (U,T,\su,E,d)$ 
  such that $T = U$,
  we have that \mbox{$\val \leq  \left(\sum_{k \in \supp(S)} \frac{S(k)}{k}\right)^{-1}$}
  where $S$ is the attack signature of $\game$.
\end{theorem}
\begin{proof}[Proof Sketch] Intuitivelly, every node $u$ has to be visited by the
  defender with propability at least $\val$ during each $d(u)$ consecutive
  steps. Hence, summing the probabilities of visiting $u$ in each of the
  steps from $1$ to $\ell = \Pi_{k \in \supp(S)}\, k$ we need to reach a
  value greater than or equal to $\val \cdot\ell/d(u)$. Summing these values
  for all nodes we have at least $\sum_{u\in U}\val \cdot \ell/d(u)$. Note
  that in each step we visit some node with probability one and so,
  the sum for all nodes and $\ell$ steps is just $\ell$. This implies the
  theorem due to $\ell\geq\sum_{u\in U}\val\cdot\ell/d(u)=\ell \cdot \val
  \cdot \sum_{k \in \supp(S)} S(k)/k$. For more details see
  Appendix~\ref{app-val-bound}.
\end{proof}

\subsection{Solving patrolling problems with a fully connected environment}
\label{sec-modular-strategies}

Let $\game = (U,T,\su,E, d)$ be a patrolling problem where
$T = U$ and $E = U \times U$, and let $S$ be the signature of $\game$. Recall the notion of
modular strategy and the associated decomposition principle introduced in
Section~\ref{sec-intro}. In particular, recall that a $d$-modular strategy
$\sigma$ for $\game$ is fully represented by probability distributions 
$\mu_0,\ldots,\mu_{d-1}$ over $U$ such that $\sigma(h) = \mu_i$ where 
$i = |h|~\mathrm{mod}~d$. 

We start by considering the case when $\game$ has $n$ nodes with the same
attack length~$d$. Since we aim at developing a strategy synthesis 
algorithm \emph{polynomial in $\size{S}$}, we need to invent a compact representation of modular strategies which is sufficiently
expressive for our purposes. 
We assume that the nodes of $U$ are 
indexed by numbers from $1$ to $|U|$, and we use $\U{i,N}$ to denote 
the subset of $U$ consisting of $N$ subsequent nodes starting from $i$,
i.e., all $u_{\ell}$ where $i \leq \ell < i+N$ and $1 \leq i \leq i+N-1 \leq |U|$.
Let us consider the
class of expressions determined by the following abstract syntax equation:
\[
\theta ~~::=~~   \Circle(\U{i,N},M,L)  ~~\mid~~ \theta_1;\theta_2 ~~\mid~~ \nu_p[\theta_1,\theta_2]
\]
Here, $M,L \in \Nset$ such that $M$ divides $N$, and $p$ ranges over a countable set of variables $\Var$. Assuming some valuation $\alpha: \Var \rightarrow [0,1]$, every expression $\theta$ determines a modular strategy for $U$ defined inductively as follows: 
$\Circle(\U{i,N},M,L)$ is a modular strategy which splits
$\U{i,N}$ into pairwise disjoint subsets of size~$M$
and then ``walks around'' these sets $L$ times,
$\theta_1;\theta_2$ is a modular strategy which ``sequentially alternates'' between $\theta_1$ and $\theta_2$, and $\nu_p[\theta_1,\theta_2]$ is a strategy
which ``composes'' $\theta_1$ and $\theta_2$
using the distribution $(1-\alpha(p),\alpha(p))$.
A detailed description of the semantics is given in 
Appendix~\ref{app-algorithm}.

Our strategy synthesis algorithm is a recursive procedure
$\Def$ which inputs a triple $(\U{i,N},D,e)$, where $\U{i,N}$ is the set of nodes to be defended, 
$D$ is the number of steps available for defending $\U{i,N}$, and $e$ is an expression
which represents the ``weight'' of the constructed defending strategy
in the final distribution $\nu$. The procedure outputs a pair $(\theta,V)$ where
$\theta$ is an expression specifying a $D$-modular strategy for $\U{i,N}$, and $V$ is 
an arithmetic expression representing the guaranteed ``coverage'' of the targets in $\U{i,N}$ when using $\theta$ with 
the weight~$e$. 
As a side effect, the function $\Def$ may produce  
equations for the variables that are employed in
symbolic strategy compositions of the form $\nu_p[\theta_1,\theta_2]$. The algorithm is invoked by 
$\Def(\U{1,|U|},d,1)$, and the system of equations is initially empty. The recursion is stopped when $D$ divides
$N$ or $N$ divides $D$, and in these cases $\Def$ provably
produces strategies that achieve the best coverage
for every value of~$e$. In the other cases, $\Def$ proceeds
recursively by splitting either the set of
nodes or the number of steps available to protect the nodes. In both cases, $\Def$ tries to exploit
the available resources in the best possible way.
A full description is given in Appendix~\ref{app-algorithm}.  
At the very end, we obtain a $d$-modular strategy $\sigma$ for $\game$
specified by an expression $\theta$ whose size is polynomial in $\size{S}$,
an expression $V$ which represents $\val(\sigma)$, and we also obtain a system of polynomial equations for the variables 
which parameterize $\theta$ and $V$. The system has a unique solution in $[0,1]^k$ (where $k$ is the number of variables) that
corresponds to the intended valuation. The size of $k$ can be, for given $n > d$, computed as follows: we put $n_0 = n$ and $d_0 =d$,
and then $n_{i+1} = n_i~\mathrm{mod}~d_i$ and $d_{i+1} = d_i~\mathrm{mod}~n_{i+1}$. The number
of variables for $n$ and $d$ is equal to the least index $j$ such that $d_j$ divides $n_j$.
In particular, if $d$ divides $n$, there is no variable at all, and our algorithm immediately
produces a strategy which achieves the value $d/n$, which is optimal by Theorem~\ref{thm-upper}.
As an example of a ``hard'' instance, consider $n = 709793170386861531$ and $d = 37248973638339152$, which requires $30$
variables and equations. The solution (producing $\val(\sigma) = 0.05247471678$) can be computed by Maple
in fractions of a second.
It has been recently proved by Lamser \cite{Lamser:BCthesis} that our algorithm produces na optimal
strategy also when $d=2$ (for arbitrary $n$), which includes the example of 
Fig.~\ref{fig-example}~(left). Since the algorithm seems to exploit
the available resources optimally, we conjecture that it actually outputs an optimal strategy for all parameters. 

To solve a patrolling problem with a general signature $S$, we simply split the nodes into
disjoint subsets according to their attack lengths, solve these subproblems by the above algorithm,
and then compose the modular strategies so that all nodes are defended equally well. One can easily check
that if $S$ is well formed, this leads to a strategy whose value matches the bound of Theorem~\ref{thm-upper}.
Thus, we obtain the following:
 
\begin{theorem}
\label{thm:well-formed-optimal}
  Let $\game$ be a patrolling problem with $T=U$, a fully connected environment, and a well formed
  signature~$S$. Then there is an optimal modular strategy $\sigma$ computable in time polynomial
  in $\size{S}$. 
\end{theorem}

\subsection{A characterization of sufficiently connected environments}
\label{sec-well-formed}

\noindent
For the rest of this subsection, we fix a patrolling problem
$\game = (U,T,\su,E,d)$ with $T =U$ and a well-formed signature~$S$.
We classify the conditions under which
$E$ is sufficiently connected (recall that
$E$ is sufficiently connected iff the value for $\game$  is the same as the value for $\game$ when $E$ is
replaced with the fully connected environment
$U \times U$. Let $\ch_S$ be a digraph with vertex labelling $d$ constructed as follows:
\begin{itemize}
\item For all $k \in \supp(S)$, $i \in \{0,\ldots,k{-}1\}$, and 
   $j \in \{1,\ldots,S(k)/k\}$, we add a fresh vertex $v_k[i,j]$
   and set $d(v_k[i,j]):=k$. Hence,
   $\ch_S$ has exactly $\sum_{k \in \supp(S)} S(k)$ vertices.
\item For every pair of vertices $v_k[i,j]$ and $v_{k'}[i',j']$,
   there is an arc from  $v_k[i,j]$ to $v_{k'}[i',j']$ in $\ch_S$ iff
   there is some $0 \leq \ell < k\cdot k'$ such that 
   $i = \ell \,\mathit{mod}\, k$ and 
   $i' = (\ell {+} 1) \,\mathit{mod}\, k'$. 
\end{itemize}
Note that $\ch_S$ is computable in polynomial time. We prove the following: 

\begin{theorem}
\label{thm-subdigraph} 
  Let $\game = (U,T,\su,E,d)$ be a patrolling problem 
  such that $T=U$ and the signature~$S$ of $\game$ is
  well formed. Then $E$ is sufficiently connected iff
  $(U,E)$ contains a subdigraph $H$ which is $d$-preserving isomorphic to~$\ch_s$
  (i.e., if $x$ of $H$ is mapped to $y$ of $\ch_S$ then $d(x)=d(y)$).
\end{theorem}

\noindent
The ``if'' part of Theorem~\ref{thm-subdigraph} is trivial, because if
$(U,E)$ contains a subdigraph $\ch_s$, then we can implement the optimal
modular strategy constructed by the algorithm of
Subsection~\ref{sec-modular-strategies}.  The ``only if'' part is more
challenging. The crucial observation is that the defender is not allowed to visit any target $u$ twice within $d(u)$ steps whenever she is
aiming to reach the bound of Theorem~\ref{thm-upper}.
 The underlying observations also reveals that \emph{every} optimal
strategy $\sigma$ starts to
behave like the strategy $\sigma^*$ after every history which visits all
nodes.  Hence, the strategy $\sigma^*$ does \emph{not} belong to
$\bigcup_{j=1}^\infty \Sigma^j$, except for some trivial cases (see
Section~\ref{sec-intro}). A proof of Theorem~\ref{thm-subdigraph} is given
in Appendix~\ref{app-subdigraph}.
An immediate consequence of Theorem~\ref{thm-subdigraph} is that
the problem whether a environment $E$ is sufficiently
connected is in \textbf{NP}. We complement this by a matching
lower bound in the following theorem
with a full proof in Appendix~\ref{app-HAM}.

\begin{theorem}
\label{thm-connected}
  The problem whether the environment of a given
  patrolling problem $\game = (U,T,\su,E,d))$,  such that 
  $T =U$ and  the signature $S$ of $\game$ is well formed,
  is sufficiently connected, is \mbox{\textbf{NP}-complete}. Further,
  this problem is \mbox{\textbf{NP}-complete} even for
  a subclass of patrolling problems such that $\supp(S) = \{k\}$, where 
  $k \geq 3$ is a fixed constant. For a subclass of patrolling
  problems where $\supp(S) = \{2\}$, the problem is solvable in 
  polynomial time.
\end{theorem}

\section{Open problems}

\noindent
Our proof of the existence of an optimal defender's strategy
(Theorem~\ref{thm:optimal}) does not allow to conclude
anything about the \emph{structure} of optimal strategies. 
One is tempted to expect that optimal strategies are in some sense ``regular'' and require only finite-memory,
but our present understanding does not allow to prove
this conjecture. Another challenge it to lift the
presented compositional technique to a more general class of patrolling games (such results would have a considerable
practical impact). Finally, the question  whether the algorithm of Section~\ref{sec-modular-strategies} produces an optimal strategy for all inputs is also interesting 
but left open.

\bibliographystyle{abbrv}

\clearpage
\appendix

\section{Detailed definitions for appendices}
\label{app-defs}

We use $\Nset$ and $\Nset_0$ to denote the sets of positive and
non-negative integers, respectively.
The sets of all finite and infinite words over a given alphabet $\Gamma$
are denoted by $\Gamma^*$ and $\Gamma^{\omega}$, respectively.  We write
$\varepsilon$ for the empty word. The length of a given $w\in \Gamma^*\cup
\Gamma^{\omega}$ is denoted by $|w|$, where the length of an infinite word
is $\infty$. We denote by $\Gamma^{\leq k}$ the set of all words $w\in \Gamma^*$
satisfying $|w|\leq k$.  The last letter of
a finite non-empty word $w$ is denoted by $\last(w)$. Given a 
(finite or infinite) 
word $w$ over $\Gamma$, the individual
letters of $w$ are denoted by $w_0 w_1\cdots$. Given two words
$w,w'\in \Gamma^*\cup \Gamma^{\omega}$ we write $w\preceq w'$ whenever $w$ is a
prefix of $w'$, i.e., whenever there exists a word $w''\in \Gamma^*\cup
\Gamma^{\omega}$ such that $w'=w w''$. Further, we write $w\prec w'$ whenever
$w\preceq w'$ and $w\not = w'$.

Given a finite or countably infinite set $A$, a \emph{probability
distribution} over $A$ is a function 
\mbox{$\delta : A \rightarrow [0,1]$} such that 
$\sum_{a\in \supp(\delta)} \delta(a)=1$. The \emph{support} of $\delta$
is the set $\supp(\delta) = \{a \in A \mid \delta(a) \neq 0\}$. 
We use $\Delta(A)$ to denote the set of all distributions over~$A$.
A distribution $\delta \in \Delta(A)$ is \emph{positive} if 
$\delta(a) > 0$ for every $a \in A$, and \emph{rational} if $\delta(a)$
is rational for every $a \in A$. 

\begin{definition} 
  A \emph{patrolling problem} is a triple $\game = (U,T,\su,E,d)$ where
  $U$ is a finite set of \emph{nodes}, $T \subseteq U$
  is a set of \emph{targets}, $\su \in T$ is the 
  \emph{initial target}, $E \subseteq U \times U$ is
  an \emph{environment}, and  $d : T \rightarrow \Nset$ assigns to 
  each target the associated \emph{attack length}. 
  The \emph{attack signature} of $\game$ is a
  function $\sig : \Nset \rightarrow \Nset_0$ where
  $\sig(k)$ is the cardinality of $\{u \in U \mid d(u) = k\}$. 
  We use $\supp(S)$ to denote the set $\{k \in \Nset \mid S(k) \neq 0\}$.
  We say that $S$ is \emph{well formed} if $k$ divides $S(k)$ for every
  $k \in \Nset$. By $\dmax$ we denote $\max_{u\in U} \{d(u)\}$.
\end{definition}

\noindent
Let $\game = (U,T,\su,E, d)$ be a patrolling problem. We say that $E$ is \emph{fully connected} if $E = U \times U$. Given a node $u \in U$, we denote by $\success(u)$ the set $\{u'\in
U \mid (u,u')\in E\}$ of all successors of $u$.
A \emph{path} is a finite or infinite word $w\in U^*\cup U^{\omega}$
such that $(w_i,w_{i+1})\in E$ for every $0\leq i< |w|$.  A
\emph{history} is a finite non-empty path, and a \emph{run} is an infinite
path. The sets of all histories and runs are denoted by $\histories$
and $\runs$, respectively. Given a set of histories 
$H \subseteq \histories$, we use $\runs(H)$ to denote the set of all runs
$\omega$ such that $w\preceq \omega$ for some $w\in H$ (when $H = \{h\}$,
we write $\runs(h)$ instead of $\runs(\{h\})$).

\begin{definition}
  A \emph{defender's strategy} is a function 
  $\sigma : \histories \rightarrow \dist(U)$ 
  such that $\supp(\sigma(h))\subseteq \success(\last(h))$ for every
  $h\in \histories$. The set of
  all defender's strategies is denoted by $\Sigma$.

  An \emph{attacker's strategy} is a function 
  $\pi : \histories \rightarrow T \cup \{\bot\}$ such that 
  whenever $\pi(h) \neq {\bot}$, then for all $h'\prec
  h$ we have that $\pi(h')={\bot}$. We denote by $\Pi$ the set of all
  attacker's strategies.
\end{definition}

\noindent
Intuitively, given a history $h$, the defender chooses the next node
randomly according to the distribution $\sigma(h)$, and the attacker
either attacks a node $u\in T$ ($\pi(h)=u$), or waits ($\pi(h)=\bot$).
Note that the attacker can choose to attack only once during a play,
and also note that he cannot randomize. This is because randomization
does not help the attacker to decrease the Stackelberg value,
and hence we can safely adopt this restriction from the very
beginning.

For a given strategy $\sigma \in \Sigma$, we define the set
$\shistories \subseteq \histories$ of \emph{relevant} histories,
consisting of all $h \in \histories$ such that
for all $h' \in \histories$ and $u \in U$ where $h'u \preceq h$ we
have that $\sigma(h')(u) > 0$. Note that a defender's
strategy $\sigma$ determines a unique probability space over all
infinite paths initiated in a given $u \in U$ in the standard way
(see, e.g., \cite{Chung:book}), and we use $\calP^\sigma_u$ to 
denote the associated probability measure.

Given an attacker's strategy $\pi$, we say that a run $w$ 
\emph{contains a successful attack} if there
exist a finite prefix $h$ of $w$ and a node $u \in T$ such that
$\pi(h) = u$ and $u$ is not among the first $d$ nodes visited by $w$
after the prefix $h$.  For every node $u \in U$, we use $\Defended_u
[\pi]$ to denote the set of all \emph{defended} runs initiated in $u$
that do not contain a successful attack. Hence,
$\calP^\sigma_u(\Defended_u[\pi])$ is the probability of all runs
initiated in~$u$ that are defended when the defender uses the strategy
$\sigma$ and the attacker uses the strategy~$\pi$. We omit the
subscript $u$ in $\calP^\sigma_u$ and $\Defended_u[\pi]$ when $u = \su$.

\begin{definition}
  For all $u \in U$ and $\sigma\in \Sigma$, we denote by 
  $\val_u(\sigma)$ the \emph{value
    of $\sigma$} defined by $ \val_u(\sigma)=\inf_{\pi \in \Pi}
  \calP^{\sigma}_u(\Defended_u[\pi])$. The \emph{Stackelberg value}
  of $u$ is defined as $\val_u = \sup_{\sigma \in \Sigma}
  \val_u(\sigma)$.  A defender's strategy $\sigma^*$ is 
  \emph{optimal in $u$} if $\val_u(\sigma^*) = \val_u$. 
  The value of $\su$ is denoted by $\val$, and a strategy
  which is optimal in $\su$ is called just \emph{optimal}.
\end{definition}
At some places, we consider strategies obtained by ``forgetting''
some initial prefix of the history. Formally, for all $h \in \histories$
and a strategy $\theta$ of the defender/attacker, we define a strategy $\theta_h$
by $\theta_h(u h') = \theta(h h')$ for every $u\in U$ and $h'\in
H$. Note that $\sigma_h$ behaves similarly for all initial
  nodes. We are typically interested in its behavior starting in $\last(h)$,
  which corresponds to behavior of $\sigma$ when started at
  $h$.
  
  In what follows, we also use the notion of an \emph{immediate attack value}.
 Given a defender's strategy $\sigma$, a history $h\in \histories(\sigma)$, and a node $u\in U$, we define $\aval_h(\sigma,u)$ to be the probability of reaching $u$ from $\last(h)$ in at least one and at most $d(u)$ steps using the strategy $\sigma_h$. Intuitively, $\aval_h(\sigma,u)$ is the probability of defending $u$ assuming that the attack on $u$ starts after the history $h$, i.e., $\pi(h)=u$. It is easy to see that
 \[
 \calP^{\sigma}(\Defended[\pi])  =  \sum_{\substack{h\in \histories(\sigma)\\ \pi(h)\not = \bot}} \calP^{\sigma}(h)\cdot \aval_h(\sigma,\pi(h))
 \]

\section{The existence of an optimal defender's strategy}
\label{app-optimal}

\begin{reftheorem}{thm:optimal}
  For every patrolling problem $\game = (U,T,\su,E,d)$ there exists an optimal defender's 
  strategy.
\end{reftheorem}
\begin{proof}
We construct an optimal strategy $\sigma^*$ as a point-wise limit of a sequence $\sigma^1,\sigma^2,\ldots$ of strategies where each $\sigma^k$ is $1/k$-optimal. More precisely, we prove the following.
\begin{claim}{}
There is a sequence of defender's strategies $\sigma^1,\sigma^2,\ldots$ and a defender's strategy $\sigma^*$ such that 
\begin{itemize}
\item each $\sigma^i$ is $1/i$-optimal, i.e., $\val(\sigma^i)\geq \val-1/i$,
\item for every $h\in \histories(\sigma)$ and every $u\in U$ we have that $\lim_{i\rightarrow \infty} \sigma^i(h)(u)=\sigma^*(h)(u)$. (In particular, the limit exists for every $h$ and $u$.)
\end{itemize}
\end{claim}
\begin{claimproof}
Assume a lexicographical ordering $\preceq$ on histories of $\histories$. To simplify our notation, we consider an "empty" history $\epsilon$ such that $\epsilon\preceq h$ for every $h\in \histories$. 
We consider histories $h$ successively according to $\preceq$ and inductively define sequences $\sigma^{h,1},\sigma^{h,2},\ldots$ of defender's strategies so that the following holds: 
\begin{itemize}
\item[A.] each $\sigma^{h,i}$ is $1/i$-optimal,
\item[B.] $\sigma^{h,1},\sigma^{h,2},\ldots$ is a subsequence of all preceding sequences $\sigma^{h',1},\sigma^{h',2},\ldots$ for $h'\preceq h$,
\item[C.] for every ${h'\preceq h}$ the sequence of distributions $\sigma^{h,1}(h'),\sigma^{h,2}(h'),\ldots$ converges (point-wisely) to a probability distribution.
\end{itemize}
Then it suffices to put $\sigma^i=\sigma^{h,|h|}$ where $h$ is the $i$-th history according to $\preceq$, and to define $\sigma^*(h)=\lim_{i\rightarrow \infty} \sigma^i(h)$.

We define $\sigma^{h,i}$ as follows:
\begin{itemize}
\item For every $i\in \Nset$, we define $\sigma^{\epsilon,i}$ to be an arbitrary $1/i$-optimal strategy.
\item Assume that $\sigma^{h',1},\sigma^{h',2},\ldots$ has already been defined for $h'$. Consider a next history $h$ according to $\preceq$. As the space of all probability distributions on $U$ is compact, there exists a~subsequence $\sigma^{h,1},\sigma^{h,2},\ldots$ of $\sigma^{h',1},\sigma^{h',2},\ldots$ such that $\sigma^{h,i}(h)$ converges (point-wisely) to a probability distribution on $U$. 
\end{itemize}
The sequences apparently satisfy the above conditions A, B, C.
\end{claimproof}

\noindent
We prove that the defender's strategy $\sigma^*$ obtained in the above Claim is optimal.
Suppose that $\sigma^*$ is not optimal, i.e. $\val(\sigma^*)\leq\val-\delta$ for some $\delta>0$. 
Then there is an attacker's strategy $\pi$ such that $\calP^{\sigma^*}(\Defended[\pi])\leq \val-\delta/2$. For every $i\in \Nset$, let $\pi_i$ behave as $\pi$ on runs where $\pi$ attacks before $i$-th step, and do not attack at all on the rest.
\begin{claim}{}
$\lim_{i\rightarrow \infty} \calP^{\sigma^*}(\Defended[\pi_i])=\calP^{\sigma^*}(\Defended[\pi])$
\end{claim}
\begin{claimproof}
Note that
\begin{eqnarray*}
\calP^{\sigma^*}(\Defended[\pi]) & = & \sum_{\substack{h\in \histories(\sigma^*)\\ \pi(h)\not = \bot}} \calP^{\sigma^*}(h)\cdot \aval_h(\sigma^*,\pi(h)) \\
& = & \sum_{\substack{h\in \histories(\sigma^*)\\ \pi(h)\not = \bot\\ |h|\leq i}} \calP^{\sigma^*}(h)\cdot\aval_h(\sigma^*,\pi(h))+\sum_{\substack{h\in \histories(\sigma^*)\\ \pi(h)\not = \bot\\ |h|>i}} \calP^{\sigma^*}(h)\cdot\aval_h(\sigma^*,\pi(h))\\
& = & \calP^{\sigma^*}(\Defended[\pi])+p_i
\end{eqnarray*}
where $p_i$ is the probability that the the attacker starts his attack after $i$. Clearly, $p_i\rightarrow \infty$ as $i\rightarrow \infty$, which proves the claim.
\end{claimproof}

Thus for a sufficiently large $i$ we have that $\calP^{\sigma^*}(\Defended[\pi_i])\leq \val-\delta/4$.
Now observe that for all sufficiently large $k\in \Nset$ we have 
$|\calP^{\sigma^*}(\Defended[\pi_i])-\calP^{\sigma^k}(\Defended[\pi_i])|\leq \delta/8$ because  the transition probabilities determined by $\sigma^*$ and $\sigma^k$ on the first $i+\dmax$ steps are getting closer and closer with growing $k$. However, then we obtain that $\calP^{\sigma^k}(\Defended[\pi_i])|\leq \val-\delta/8$, which means that $\sigma^k$ cannot be $1/k$-optimal for large  $k$.
\end{proof}

\begin{proposition}\label{prop:optimal}
Assume that $\su$ is a target. There every optimal defender's strategy $\sigma^*$ satisfies
\begin{equation}\label{eq:atval-opt}
\inf_{h\in \histories(\sigma^*)} \ \ \min_{u\in T}\ \ \aval_h(\sigma^*,u)\quad \geq\quad \val
\end{equation}
\end{proposition}
\begin{proof}
Recall that we denote by $\val_u$ and $\val_u(\sigma)$ the values of $\game$ and of $\sigma$, resp., when $u$ is used as the initial node instead of $\su$.
It suffices to prove Proposition~\ref{prop:optimal} under the assumption
that $\val=\val_{\su}=\max_{u\in T} \val_{u}$, because then we obtain, as a~consequence, that $\val_u=\val_{\su}$ for all $u\in T$. Indeed, using $\sigma^*$, every target node has to be visited. So given $u\in T$, there is a history $h\in \histories(\sigma^*)$ such that $u=\last(h)$. However, note that~(\ref{eq:atval-opt}) holds also for $\sigma^*_h$ instead of $\sigma^*$, and thus $\val_u(\sigma^*_h)\geq \val$. As $\val_{\su}=\val$ is maximal, we obtain that $\val_u=\val_{\su}$.

So assume that $\val=\val_{\su}=\max_{u\in T} \val_{u}$. Let $\sigma^*$ be an optimal strategy. Note that $\val=\max_{u\in T}\val_u$ implies $\val_{\last(h)}(\sigma^*_h)\leq \val$ for every history $h\in \histories(\sigma)$ such that $\last(h)\in T$. We obtain that $\val_{\last(h)}(\sigma_h)\leq \val$ {\em for every history} $h\in \histories(\sigma)$ because even if $\last(h)$ is not a target, $\sigma_h$ starting in $\last(h)$ must visit a target almost surely and the attacker may wait until it happens.

We claim that $\sigma^*$ satisfies (\ref{eq:atval-opt}), i.e. that $\aval_h(\sigma^*,u)\geq\val$ for all $h\in \histories(\sigma^*)$ and all $u\in T$. Indeed, assume that 
$\aval_{\bar{h}}(\sigma^*,u) \leq\val-\delta$ for some $\delta>0$ and $\bar{h}\in \histories(\sigma^*)$ and $u\in U$. Assume, w.l.o.g., that $\sigma^*$ follows the history $\bar{h}$ with probability at least $\delta$.

Note that due to $\val_{\last(h)}(\sigma^*_{h})\leq \val$ for every $h$, the deficiency of $\sigma^*$ at $\bar{h}$ cannot be compensated on other histories.
We obtain the following: Let $A$ be the set of all histories $h'$ of length $|h|$ (i.e., in particular, $h\in A$). Then
\begin{eqnarray*}
\val(\sigma^*) & \leq  & \sum_{h'\in A} \calP^{\sigma^*}(h') \cdot \val_{\last(h')}(\sigma^*_{h'}) \\
	& = & \calP^{\sigma^*}(h) \cdot \val_{\last(h)}(\sigma^*_h) + \sum_{h'\in A\smallsetminus \{h\}} \calP^{\sigma^*}(h') \cdot \val_{\last(h')}(\sigma^*_{h'}) \\
	& \leq & \calP^{\sigma^*}(h) \cdot \val_{\last(h)}(\sigma^*_h) + \sum_{h'\in A\smallsetminus \{h\}} \calP^{\sigma^*}(h') \cdot \val \\
	& \leq & \calP^{\sigma^*}(h) \cdot\min_{u\in U}\ \aval_h(\sigma^*,u) + \sum_{h'\in A\smallsetminus \{h\}} \calP^{\sigma^*}(h') \cdot \val \\
	& \leq & \calP^{\sigma^*}(h) \cdot (\val-\delta) + \sum_{h'\in A\smallsetminus \{h\}} \calP^{\sigma^*}(h') \cdot \val \\
	& = & \val-\calP^{\sigma^*}(h) \cdot \delta\\
	& \leq & \val-\delta^2
\end{eqnarray*}
This contradicts the fact that $\sigma^*$ is optimal.

\end{proof}

\section{Computing finite-memory $\varepsilon$-optimal strategies}
\label{app-eps-opt}

\subsection{Proposition~\ref{cor:eps-opt-discr}}
Let us fix a patrolling problem $\game = (U,T,\su,E,d)$.
\begin{proposition}\label{cor:eps-opt-discr}
Given $\varepsilon>0$, there is an $\varepsilon$-optimal strategy
$\sigma^{\varepsilon}$ such that for every history $h$ and every $u\in U$ it
holds
\[
\sigma^{\varepsilon}(h)(u)=k\cdot \lceil (|U|\dmax)/\varepsilon\rceil^{-1}\text{ for a suitable }k\in \Nset 
\]
and
\[
\min_{u\in U}\ \ \aval_h(\sigma^\varepsilon,u)\quad \geq\quad \val-\varepsilon.
\]
\end{proposition}
\begin{proof}
Let $\sigma$ be an optimal strategy.
Let $U = \{u_1, u_2, \dots, u_{|U|}\}$ be the set of nodes of $\mathcal{G}$ and define $s=\lceil (|U|\dmax)/\varepsilon\rceil^{-1}$.

For every history $h$ and every $1 \leq i \leq |U|$, we inductively define $\sigma^{\varepsilon}(h)(u_i) = k_i \cdot s$, where $k_i$  is the largest number satisfying

$$ k_i \cdot s \leq \sigma(h)(u_i) + \sum_{j=1}^{i-1}(\sigma(h)(u_j) - \sigma^{\varepsilon}(h)(u_j))~.$$

This rounding procedure guarantees
that $\sigma^{\varepsilon}(h)$ is indeed a probability distribution over $U$, i.e. $\sum_{u\in U}\sigma^{\varepsilon}(h)(u) = 1$ (note that simple rounding would not guarantee this property).
Further, when we realize the invariant $0 \leq \sum_{j=1}^{i-1}(\sigma(h)(u_j) - \sigma^{\varepsilon}(h)(u_j)) < s$ holds for all $1 \leq i \leq |U|$, it is easy to see that $|\sigma(h)(u_i) - \sigma^{\varepsilon}(h)(u_i)| < s$, which is captured by the following claim.
\begin{claim}{A}
$|\sigma(h)(u) - \sigma^{\varepsilon}(h)(u)| < s$ for every $u\in U$.
\end{claim}

It follows from the definition of $\sigma^{\varepsilon}$ that whenever $\sigma(h)(u) = 0$, then also $\sigma^{\varepsilon}(h)(u) = 0$.
This means that any history executable using $\sigma^{\varepsilon}$ is also executable using $\sigma$, i.e. $\mathcal{H}(\sigma^{\varepsilon}) \subseteq \mathcal{H}(\sigma)$.

Now, knowing that $\aval_h(\sigma)$ is defined if $\aval_h(\sigma^{\varepsilon})$ is defined, we prove the following:

\begin{align}
\aval_h(\sigma^{\varepsilon})
  &\geq \aval_h(\sigma) - \epsilon\label{eq1:first} \\
  &\geq val_{\last(h)}(\sigma_h) - \epsilon\label{eq1:second} \\
  &\geq  val(\sigma) - \epsilon\label{eq1:third}~.
\end{align}

\begin{itemize}

	\item The inequality \eqref{eq1:third} directly follows from Proposition~\ref{prop:optimal} as $h \in \histories(\sigma)$.

	\item The inequality \eqref{eq1:second} clearly holds as forcing the attacker to attack immediately cannot decrease the value of the game.

	\item To prove the first inequality \eqref{eq1:first}, we have to analyze the impact of the rounding in the definition of $\sigma^{\varepsilon}$.

Denote by $R[\iota, h, t, k]$ the probability of reaching $t \in T$ from $last(h)$, $h \in \mathcal{H}$, in up to $k$ steps using the strategy $\iota_h$.

We prove by induction on $k$ that for all $h \in \mathcal{H}$, $t \in T$, and $k \in \mathbb{N}$ we have that

\begin{align}
R[\sigma, h, t, k] - R[\sigma^{\varepsilon}, h, t, k] \leq k|U|s~.\nonumber
\end{align}

The base case ($k = 1$) directly follows from Claim~A for all $u\in U$ and the fact that $R[\iota, h, t, 1] = \iota(h)(t)$ for every defender's strategy $\iota$.

Let us denote the difference $R[\sigma, h, t, k] - R[\sigma^{\varepsilon}, h, t, k]$ by $\Delta$.
For $k \geq 2$, we have
\begin{eqnarray}
\Delta
  &=& \sigma(h)(t) + \sum_{u\in U \smallsetminus \{t\}}\sigma(h)(u)\cdot R[\sigma, hu, t, k - 1] - \label{eq1b:first} \\
  & &\hspace{-11.2pt} -~\sigma^{\varepsilon}(h)(t) - \sum_{u\in U \smallsetminus \{t\}}\sigma^{\varepsilon}(h)(u)\cdot R[\sigma^{\varepsilon}, hu, t, k - 1]	\nonumber \\
  &\leq& s + \sum_{u\in U \smallsetminus \{t\}}(R[\sigma, hu, t, k - 1]\cdot (\sigma(h)(u) - \sigma^{\varepsilon}(h)(u)) \label{eq1b:second} +	\\
  & &\hspace{47.5pt} +~\sigma^{\varepsilon}(h)(u)\cdot (R[\sigma, hu, t, k - 1] - R[\sigma^{\varepsilon}, hu, t, k - 1])) \nonumber	 \\
  &\leq& s + \sum_{u\in U \smallsetminus \{t\}} R[\sigma, hu, t, k - 1]\cdot s + \label{eq1b:third}	\\	& &\hspace{7.6pt} +\hspace{-2pt}~\sum_{u\in U \smallsetminus \{t\}} \sigma^{\varepsilon}(h)(u)\cdot (k - 1)|U|s \nonumber	\\
	  &\leq& s + (|U|-1)s + (k - 1)|U|s	\label{eq1b:fourth} \\
  &=& k|U|s\nonumber~.
\end{eqnarray}

\begin{itemize}
	\item The equality \eqref{eq1b:first} follows from the definition of $R[\iota, h, t, k]$ as $R[\iota, h, t, k] = \iota(h)(t) + \sum_{u \in U \smallsetminus \{t\}} \iota(h)(u) \cdot R[\iota, hu, t, k - 1]$ for all $k \geq 2$. 	\item The equality \eqref{eq1b:second} is just an application of Claim~A and of the formula $ab - a'b' = b(a -  a') + a'(b - b')$.
	\item The inequality \eqref{eq1b:third} follows from Claim~A and from the induction hypothesis.
	\item The inequality \eqref{eq1b:fourth} holds because $R[\sigma, hu, t, k - 1] \leq 1$ and $\sum_{u\in U \smallsetminus \{t\}} \sigma^{\varepsilon}(h)(u) \leq 1$.
\end{itemize}
So we have that $R[\sigma, h, t, d(t)] - R[\sigma^{\varepsilon}, h, t, d(t)] \leq d|U|s$.
However, note that

$$\aval_h(\iota) = \inf_{t \in T} R[\iota, h, t, d(t)]$$

and therefore

$$\aval_h(\sigma) - \aval_h(\sigma^{\varepsilon}) \leq \dmax|U|s \leq \epsilon~.$$

\end{itemize}

\end{proof}

\subsection{Formal proof of Theorem~\ref{thm:eps-opt}}

\noindent

\noindent
In order to make lengthy computations more succinct, we use the following shorthand notation:
Given a~characteristic $c=(\croot,\cfirst,\commit)\in \Char$, we define:
\begin{itemize}
\item $c(0,u)=1$ if $u=c_{\croot}$, and $c(0,u)=0$ for all $u\not = c_{\croot}$.
\item $c(1,u)=c_{\cfirst}(u)$ for all $u\in U$.
\item $c(k,u)=c_{\commit}(k,u)$ for all $2\leq k\leq \dmax$ and $u\in T$.
\end{itemize}
Also, we use functional notation to denote vectors of characteristics (i.e. successors). That is we represent each $(c^v)_{v\in U}\in \Char^U$ as a function $\zeta:U\rightarrow \Char$ where $\zeta(v)=c^v$ for every $v\in U$.

Let us formally define the notion of successor of a characteristic. We say that $\zeta:U\rightarrow \Char$ is a~{\em successor} of $c\in \Char$ if for every $v\in U$ holds $\zeta(v)(0,v)=1$, and for every $u\in T$ and $2\leq k\leq \dmax$ holds
\[
c(k,u)=c(1,u)+\sum_{v\not = u} c(1,v)\cdot \zeta(v)(k-1,u)
\]
A set of characteristics $B\subseteq \mathcal{C}$ is {\em closed} if there is at least one $c\in B$ satisfying $c(0,\su)=1$, and every $c\in B$ has a successor $\zeta:U\rightarrow B$. 

Given a defender's strategy $\sigma$ and a~history $h$, we denote by $c[\sigma,h]$ the characteristic defined as follows: $c[\sigma,h](\last(h))=1$, and $c[\sigma,h](1,u)=\sigma(h)(u)$ for every $u\in U$, and for every $2\leq k\leq \dmax$ and $u\in T$ we define
 \[
 c[\sigma,h](k,u)=\calP^{\sigma}(\runs(\{hh'\mid \last(h')=u, 1\leq |h'|\leq k\})\mid \runs(h))
 \]
 (Intuitively, for $k\geq 1$, the value $c[\sigma,h](k,u)$ is the probability of reaching $u$ in at least one and at most $k$ steps starting with the history $h$ and using $\sigma$.) Denote by $\Char[\sigma]$ the set of all characteristics $c[\sigma,h]$ where $h\in \histories(\sigma)$.

\begin{lemma}
Given a defender's strategy $\sigma$, the set $\Char[\sigma]$ is closed.
\end{lemma}
\begin{proof}
By definition, $c[\sigma,\hat{u}](0,\hat{u})=1$. Now consider $c[\sigma,h]\in \Char[\sigma]$. Let $\xi(v)=c[\sigma,hv]$. Apparently, $\xi(v)\in \Char[\sigma]$ so it suffices to show that $\xi$ is a successor of $c[\sigma,h]$. By definition,
\[
\xi(v)(0,v)=c[\sigma,hv](0,v)=1
\]
and, clearly,
\begin{eqnarray*}
c[\sigma,h](k,u) & = & c[\sigma,h](1,u)+\sum_{v\not = u} c[\sigma,h](1,v)\cdot c[\sigma,hv](k-1,u)\\
& = & c[\sigma,h](1,u)+\sum_{v\not = u} c[\sigma,h](1,v)\cdot \xi(v)(k-1,u)\\
\end{eqnarray*}
which means that $\xi$ is a successor of $c[\sigma,h]$.
\end{proof}
\noindent
Let $C$ be a finite closed subset of $\Char$. We say that a finite-memory strategy $\sigma=(M,N,m_0,\xi)$ is {\em consistent} with $C$ if
\begin{itemize}
\item $M=C$,
\item for every $c\in C$ the function $K(c)$ defined by 
\[
K(c)(u)=N(c,u)\text{ \quad for all } u\in U
\]
is a successor of $c$.

\item $m_0=\hat{c}$ for some $\hat{c}\in C$ satisfying $\hat{c}(0,\hat{u})=1$,
\item $\xi(c,u)=c(1,u)$ for all $u\in U$.
\end{itemize}
\begin{proposition}\label{prop:closed-set-strat}
Let $C$ be a finite closed set of characteristics and assume that $\sigma$ is consistent with $C$. Then $\val(\sigma)\geq \min_{c\in C} \val(c)$.
\end{proposition}
\begin{proof}
Let us first prove that $c[\sigma,h]\in C$ for every history $h\in \histories(\sigma)$. Let us fix a history $h$.
We prove that $c[\sigma,h]=N(\hat{c},h)$, i.e. that $c[\sigma,h](k,u)=N(\hat{c},h)(k,u)$ for all $0\leq k\leq \dmax$.
It is easy to show that $N(\hat{c},h)(0,\last(h))=1$. For $k>0$ we proceed by induction on $k$.
Immediately from definitions we obtain that for every $u\in U$
\[
c[\sigma,h](1,u)=\sigma(h)(u)=\xi(N(\hat{c},h),u)=N(\hat{c},h)(1,u)
\]
Now consider $2\leq k\leq \dmax$. We have
\begin{eqnarray*}
c[\sigma,h](k,u) & = & c[\sigma,h](1,v)+\sum_{v\not = u} c[\sigma,h](1,v)\cdot c[\sigma,hv](k-1,u)\\
& = & N(\hat{c},h)(1,u)+\sum_{v=u} N(\hat{c},h)(1,v)\cdot N(\hat{c},hv)(k-1,u)\\
& = & N(\hat{c},h)(1,u)+\sum_{v=u} N(\hat{c},h)(1,v)\cdot N(N(\hat{c},h),v')(k-1,u)\\
& = & N(\hat{c},h)(1,u)+\sum_{v=u} N(\hat{c},h)(1,v)\cdot K(N(\hat{c},h))(v')(k-1,u)\\
& = & N(\hat{c},h)(k,u)
\end{eqnarray*}
Here the second equality follows by induction, the last equality follows from the fact that $K(N(\hat{c},h))$ is a successor of $N(\hat{c},h)$.
This proves that $c[\sigma,h]\in C$ for every history $h\in \histories(\sigma)$.

Now since every defender's strategy $\sigma$ satisfies
\[
\aval_{h}(\sigma,u)=c[\sigma,h](d(u),u)
\]
we obtain
\begin{eqnarray*}
\val(\sigma) & = & \inf_{\pi} \calP^{\sigma}_{\su}(\Defended[\pi]) \\
 & = & \inf_{\pi}\sum_{h\in \histories(\sigma), \pi(h)\in U} \calP(\runs(h))\cdot \aval_h(\sigma,\pi(h)) \\
 & = & \inf_{\pi}\sum_{h\in \histories(\sigma), \pi(h)\in U} \calP(\runs(h))\cdot 
 c[\sigma,h](d(\pi(h)),\pi(h)) \\
  & \geq  & \inf_{\pi}\sum_{h\in \histories(\sigma), \pi(h)\in U} \calP(\runs(h))\cdot \min_{c\in C}\val(c) \\
  & = & \min_{c\in C}\val(c)
\end{eqnarray*}
\end{proof}

\noindent
Let $\Char_{\varepsilon}$ be the set of all characteristics $c$ such that $c(k,u)$ is an integer multiple of $s^k$, here $s=\lceil |U| d/\varepsilon \rceil^{-1}$, for every $1\leq k\leq \dmax$ and every $u\in U$.
\begin{lemma}\label{lem:closed-eps}
The set $\Char_{\varepsilon}$ contains a (finite) closed subset $C$ such that $\min_{c\in C} \val(c)\geq \val-\varepsilon$.
\end{lemma}
\begin{proof}
It suffices to consider $\sigma^{\varepsilon}$ of Proposition~\ref{cor:eps-opt-discr}. Then $\Char[\sigma^{\varepsilon}]\subseteq \Char_{\varepsilon}$ is a closed subset.
\end{proof}
\noindent
Given any closed subset $C$ of $\Char_{\varepsilon}$ satisfying $\min_{c\in C} \val(c)\geq \val-\varepsilon$, we obtain, via Proposition~\ref{prop:closed-set-strat}, a finite-memory $\varepsilon$-optimal strategy. So it remains to give an algorithm for computing such a closed subset $C$.

\subsection*{The Algorithm}
The following procedure computes a closed subset $C$ of $\Char_{\varepsilon}$ which maximizes $\min_{c\in C}\val(c)$ among all closed subsets of $\Char_{\varepsilon}$ (so in particular, satisfies the desired bound $\min_{c\in C}\val(c)\geq \val-\varepsilon$).

Let $c_1,\ldots,c_n$ be all characteristics of $\Char_{\varepsilon}$ ordered in such a way that for arbitrary $c_i, c_j$ we have that $\val(c_i)\leq \val(c_j)$ implies $i\leq j$.
The following procedure maintains the invariant that $A=\{c_1,\ldots,c_k\}$ for some $k\geq 0$ and computes the desired closed set $C$ :
\begin{itemize}
\item[1.] Initialize $A:=\{c_1\}$.
\item[2.] Compute a closed subset of $A$, or indicate that $A$ does not contain a closed subset as follows:
	\begin{itemize}
	\item[a.] Initialize $B:=A$,
	\item[b.] compute $B'$ as the set of all $c\in B$ that have successors in $B$, 
	\item[c.] depending on $B'$ do:
		\begin{itemize}
		\item if either $B'=\emptyset$, or there is no $c\in B'$ such that $c(0,\hat{u})=1$, then indicate that there is no closed subset of $A$ (and proceed to 3.),
		\item else, if $B=B'$, then return $B$ as a closed subset of $A$ (and proceed to 3.),
		\item else assign $B:=B'$ and go to b.
		\end{itemize}
	\end{itemize}
\item[3.] If $A=\{c_1,\ldots,c_k\}$ does not contain a closed subset, then add $c_{k+1}$ to $A$ and go to 2., else return the closed subset as a result.
\end{itemize}
\subsubsection*{Correctness}
In step 2., the algorithm computes the greatest closed subset of $A$ using a straightforward iterative algorithm. As the characteristics are added to $A$ in the order of non-decreasing value and there exists a closed subset $C$ of $\Char_{\varepsilon}$ satisfying $\min_{c\in C} \val(c)\geq \val-\varepsilon$ (due to Lemma~\ref{lem:closed-eps}), a subset $C'$ satisfying $\min_{c\in C'} \val(c)\geq \val-\varepsilon$ is computed when $C\subseteq A$ for the first time.

\subsubsection*{Complexity}
Let us denote by $\Theta$ the size of $\Char_{\varepsilon}$. It is straightforward to show that $\Theta\in \left(\frac{|U|\dmax}{\varepsilon}\right)^{\calO(|U|\dmax^2)}$. Now the computation in step 2. b. takes time in $\Theta^{\calO(|U|)}$ (for every characteristic of $B$ one has to check all possible successors, i.e. vectors of the form $\zeta:U\rightarrow B$). The whole algorithm iterates at most $\Theta$ times through 1. -- 3. (a characteristic is added to $A$ in every iteration except the last one). So the total complexity is at most 
\[
\Theta^{\calO(|U|)}= \left(\frac{|U|\dmax}{\varepsilon}\right)^{\calO(|U|^2\dmax^2)}
\]

\section{A bound on the Stackelberg value}
\label{app-val-bound}

Using the arguments of the proof of Proposition~\ref{prop:optimal}, the
following proposition can be shown. 

\begin{proposition}
\label{prop-abafy}
Let $\game = (U,T,\su,E,d)$ be a patrolling problem. Further, let $\sigma$ be an optimal defender's
strategy and $h \in \shistories$. Then $\val_{\last(h)}(\sigma_h) =
\val(\sigma) = \val$.
\end{proposition}

\noindent
Note that Proposition~\ref{prop-abafy} cannot be generalized to non-optimal
strategies, i.e., for a given non-optimal $\sigma$ and $h \in \shistories$
we do \emph{not} necessarily have that $\val_{\last(h)}(\sigma_h) =
\val(\sigma)$ (a counterexample is easy to find).

\begin{reftheorem}{thm-upper}
  For every patrolling problem $\game = (U,T,\su,E,d)$ 
  such that $T = U$,
  we have that $\val \leq  \left(\sum_{k \in \supp(S)} \frac{S(k)}{k}\right)^{-1}$
  where $S$ is the attack signature of $\game$.
\end{reftheorem}
\begin{proof}
  Let $\sigma$ be an optimal defender's strategy. For all 
  $h \in \shistories$ and $i \in \Nset_0$,
  let $\Node_{h,i} : \runs(h) \rightarrow U$ be a function which to every
  run $hw \in \runs(h)$ assigns the node $w_i$.
  Further, let  $\mu_{h,i} \in \Delta(U)$ be a distribution defined by 
  $\mu_{h,i}(u) = \calP^\sigma (\Node_{h,i} {=} u)/\calP^\sigma(\runs(h))$. 

  First, we show that for all $u \in U$ and $i \in \Nset_0$ we have 
  that $\sum_{j = i}^{i+d(u)-1} \mu_{\su,j}(u) \geq \val$. Let us fix some 
  $u \in U$ and $i \in \Nset_0$, and let $\histories_{i}(\sigma)$ be the
  set of all $h \in \shistories$ such that $|h| = i$.
  For every $h \in \histories_{i}(\sigma)$, consider an attacker's 
  strategy $\pi$ such that $\pi(h) = u$. Due to 
  Proposition~\ref{prop-abafy}, we have that
  $\val_{\last(h)}(\sigma_h) = \val$, which means that
  $\calP^\sigma_{\last(h)}(\Defended_{\last(h)}[\pi_h])$ is at least $\val$.
  Obviously, 
  $\calP^\sigma_{\last(h)}(\Defended_{\last(h)}[\pi_h]) \ \leq \ 
   \sum_{j = 0}^{d(u)-1} \mu_{h,j}(u)$. Thus, we obtain
  $\sum_{j = 0}^{d(u)-1} \mu_{h,j}(u) \geq \val$. Now it suffices to
  realize 
  \[
    \sum_{j = i}^{i+d(u)-1} \mu_{\su,j}(n) \quad = \quad 
    \sum_{h \in \histories_{i}(\sigma)} \left (
    \calP^{\sigma}(\runs(h)) \cdot \sum_{j = 0}^{d(u)-1} \mu_{h,j}(u) \right )
    \quad \geq \quad  
    \val \cdot \sum_{h \in \histories_{i}(\sigma)} \calP^{\sigma}(\runs(h))
    \quad = \quad \val \,.
  \]

  Now we can continue with the main proof. Let $\ell = \Pi_{k \in \supp(S)}\, k$.
  Since $\sum_{j = i}^{i+d(u)-1} \mu_{\su,j}(u) \geq \val$ for all $u \in U$
  and $i \in \Nset_0$ (see above), we immediately obtain
  $\sum_{j = 0}^{\ell-1} \mu_{\su,j}(u) \geq \val \cdot \frac{\ell}{d(u)}$.
  Hence,
  \[
    \ell \quad = \quad \sum_{j = 0}^{\ell-1} \sum_{u \in U} \mu_{\su,j}(u)
    \quad = \quad
    \sum_{u \in U} \sum_{j = 0}^{\ell-1} \mu_{\su,j}(u) 
    \quad \geq \quad \sum_{u \in U} \val \cdot \frac{\ell}{d(u)} 
    \quad = \quad \val \cdot \sum_{k \in \supp(S)} \frac{S(k) \cdot \ell}{k}
  \]
  Thus, we get 
  $\val \leq  \left(\sum_{k \in \supp(S)} \frac{S(k)}{k}\right)^{-1}$
  as desired.
\end{proof}

\section{Solving patrolling problems with a fully connected environment}
\label{app-algorithm}

Let $\game = (U,T\su,E,d)$ be a patrolling problem where
$T = U$, $E = U \times U$, and  let $S$ be the signature of $\game$.  
We start by defining the semantics for the
``strategy expressions'' introduced in 
Section~\ref{sec-modular-strategies} precisely.

\begin{itemize}
	\item  $\Circle(\U{i,N},M,L)$ denotes the $c$-modular strategy where $c = L \cdot (N/M)$ such that  
	the distribution $\mu_\ell$, where $0 \leq \ell < c$, selects uniformly
	among the elements of $\U{i+\hat{\ell} \cdot M, M}$ where $\hat{\ell} = \ell~\mathrm{mod}~(N/M)$. 
        In other words, $\Circle(\U{i,N},M,L)$ is a strategy which splits $\U{i,N}$ into pairwise disjoint subsets of size~$M$
	and then ``walks around'' these sets $L$ times (actually, $\Circle(\U{i,N},M,L)$ can also be seen as 
	$(N/M)$-modular strategy, but for technical reasons we prefer to interpret it as a $c$-modular strategy).
	\item if $\theta_1$ and $\theta_2$ denote $c_1$-modular and $c_2$-modular strategies
	with underlying distributions $\mu_0^1,\ldots,\mu_{c_1-1}^1$ and $\mu_0^2,\ldots,\mu_{c_2-1}^2$, respectively,
	then $\theta_1;\theta_2$ denotes the $c_1{+}c_2$-modular strategy with the underlying distributions
	$\mu_0^1,\ldots,\mu_{c_1-1}^1,\mu_0^2,\ldots,\mu_{c_2-1}^2$.
	\item if $\theta_1$ and $\theta_2$ denote $c$-modular strategies
	with underlying distributions $\mu_0^1,\ldots,\mu_{c-1}^1$ and $\mu_0^2,\ldots,\mu_{c-1}^2$, then
	Then $\nu_p[\theta_1,\theta_2]$ denotes the $c$-modular strategy with the underlying distributions
	$\mu_0,\ldots,\mu_{c-1}$, where $\mu_i = (1-\alpha(p)) \cdot \mu_i^1 + \alpha(p) \cdot \mu_i^2$ for all
	$0 \leq i < c$.
\end{itemize}

Now we give a detailed description of the algorithm of Section~\ref{sec-modular-strategies}.
We construct a recursive function 
$\Def$ which inputs a triple $(\U{i,N},D,e)$, where $\U{i,N}$ is the set of nodes to be defended, 
$D$ is the number of steps available for defending $\U{i,N}$, and $e$ is an expression
which represents the ``weight'' of the constructed defending strategy
in the final distribution $\nu$. The procedure outputs a pair $(\theta,V)$ where
$\theta$ is an expression specifying a $D$-modular strategy for $\U{i,N}$, and $V$ is 
an arithmetic expression representing the guaranteed ``coverage'' of the targets in $\U{i,N}$ when using $\theta$ with 
the weight~$e$. 
As a side effect, the function $\Def$ may produce  
equations for the employed variables. The algorithm is invoked by 
$\Def(\U{1,|U|},d,1)$, and the system of equations is initially empty. 
A call $\Def(\U{i,N},D,e)$ is processed as follows:
\begin{itemize}
	\item If $D \mid N$ and $N = k \cdot D$, then 
	$\theta = \Circle(\U{i,N},k,1)$. Observe that every node of $\U{i,N}$ 
	is visited at most once in $D$ steps, and this happens with probability $D/N$. 
	If the weight of $\theta$ is $e$, then this probability becomes $e\cdot(D/N)$ (since
	$N$ and $D$ are constants, the expression $V = e\cdot(D/N)$ is parameterized just by 
	the variables of $e$). Hence, the function returns the pair $(\theta,V)$.  
	\item If $N \mid D$ and $D = k \cdot N$, then 
	$\theta = \Circle(\U{i,N},1,k)$. Every node of $\U{i,N}$
	is visited precisely $k$ times in $D$ steps. If the weight of $\theta$ is $e$, then the probability 
	of visiting a given node in $D$ steps is $V =  1- (1-e)^k$. 
	The function returns the pair $(\theta,V)$.
	\item If $N > D$, $D \nmid N$, and $N = k \cdot D + c$ where $1 \leq c < D$, then
	we split the set $\U{i,N}$ into disjoint subsets $\U{i,k \cdot D}$ and $\U{i + k \cdot D,c}$ with precisely 
	$k \cdot D$ and $c$ elements, respectively. Then, we pick a fresh variable 
	$p$ and issue two recursive calls:
	\[
	(\theta_1,V_1) = \Def(\U{i,k \cdot D},D,(1-p)\cdot e),\hspace*{2em}
	(\theta_2,V_2) = \Def(\U{i + k \cdot D,c},D,p \cdot e)
	\]
	The set of equations is enriched by $V_1 = V_2$. That is, we require that $p$ is chosen so that
	the nodes of $\U{i,k \cdot D}$ and $\U{i + k \cdot D,c}$ are protected equally well. Then, we put
	$\theta = \nu_p[\theta_1,\theta_2]$ and we set $V = V_1$. The function returns 
	the pair $(\theta,V)$.
	\item Finally, if $D > N$, $N \nmid D$, and $D = k \cdot N + c$ where $1 \leq c < N$, we issue two
	recursive calls:
	\[ 
	(\theta_1,V_1) = \Def(\U{i,N},k\cdot N, e),\hspace*{2em}
	(\theta_2,V_2) = \Def(\U{i,N},c,e)
	\]
	This is perhaps the most subtle part of our algorithm. Here we do not split the set $\U{i,N}$, but 
	the number of steps available to protect~$\U{i,N}$. Intuitively, the constructed strategy $\theta$ first
	tries to loop over the targets of $\U{i,N}$ as long as possible (i.e., for the first $k\cdot N$ steps).
	This is what $\theta_1$ does. Then, $\theta$ tries to exploit the remaining $c$ steps in the 
	best possible way, i.e., by employing $\theta_2$. That is, we put $\theta = \theta_1;\theta_2$.
	If the weight of $\theta$ is $e$,
	then the targets of $\U{i,N}$ are protected with probability at least $V = 1 - (1-V_1)(1-V_2)$. The function 
	returns the pair $(\theta,V)$. 
\end{itemize}

\section{The existence of a characteristic subdigraph}
\label{app-subdigraph}

In this section we prove the non-trivial direction of 
Theorem~\ref{thm-subdigraph}, i.e., we show that 
if $\game = (U,T,\su,E,d)$ is a patrolling problem
with $T = U$, a well formed attack signature $S$,
and a sufficiently connected environment, then $\ch_S$ is
($d$-preserving isomorphic to) a subdigraph of $(U,E)$.
 
Let us assume that $E$ is sufficiently connected, and let
$\sigma$ be a defender's strategy for $\game$ such
that $\val(\sigma) = \left(\sum_{k \in \supp(S)} \frac{S(k)}{k}\right)^{-1}$.
Due to Theorem~\ref{thm-upper}, we obtain that $\sigma$ is \emph{optimal},
i.e., $\val = \val(\sigma)$, and hence we can apply 
Proposition~\ref{prop-abafy}~to~$\sigma$.

We reuse the notation introduced in the proof of 
Theorem~\ref{thm-upper}. In particular, for all 
$h \in \shistories$ and $i \in \Nset_0$, we use 
$\Node_{h,i} : \runs(h) \rightarrow U$ to denote a function which to every
run $hw \in \runs(h)$ assigns the node $w_i$.
Further, we use  $\mu_{h,i} \in \Delta(U)$ to denote a distribution 
defined by 
$\mu_{h,i}(u) = \calP^\sigma (\Node_{h,i} {=} u)/\calP^\sigma(\runs(h))$. 
We start by realizing the following:

\begin{lemma}
\label{lem-mu-val}
  For all  $h \in \shistories$ and $u \in U$, we have that
  $\sum_{i=0}^{d(u)-1} \mu_{h,i}(u) = \val$.
\end{lemma}
\begin{proof}
  For all $h \in \shistories$ and $u \in U$ we have that 
  \[
    \sum_{i=0}^{d(u)-1} \mu_{h,i}(u) \quad \geq \quad 
    \val_{\last(h)}(\sigma_h) \quad = \quad \val
  \]
  where the last equality is due to Proposition~\ref{prop-abafy}.
  Now suppose that there exist some $h \in \shistories$ and $u \in U$
  such that $\sum_{i=0}^{d(u)-1} \mu_{h,i}(u) > \val$. Let 
  $\ell = \Pi_{k \in \supp(S)}\, k$. For every $k \in \supp(S)$, we put
  \[
     \alpha[k] \quad = \quad \sum_{u \in U,\ d(u)=k} \ 
                             \sum_{i=0}^{\ell-1}\ \mu_{h,i}(u)\,.
  \]
  Obviously, $\sum_{k \in \supp(S)} \alpha[k] = \ell$. Further, for every
  $k \in \supp(S)$ we have that 
  $\alpha[k] \geq \val \cdot S(k) \cdot \frac{\ell}{k}$, because otherwise
  there inevitably exists some $0\leq i < \ell-k$ and $u \in U$ such that
  $d(u) = k$ and $\sum_{j=i}^{i+d(u)-1} \mu_{h,i}(u) < \val$, which means that
  there exists $hh' \in \shistories$ such that $|h'| = i$ and
  $\sum_{j=0}^{d(u)-1} \mu_{hh',j}(u) < \val$. Since 
  $\val_{\last(hh')}(\sigma_{hh'}) = \val$ by 
  Proposition~\ref{prop-abafy}, we have a contradiction.

  Since $\alpha[k] \geq \val \cdot S(k) \cdot \frac{\ell}{k}$ for all
  $k \in \supp(S)$ and $\sum_{k \in \supp(S)} \alpha[k] = \ell$, we obtain 
  that $\alpha[k] = \val \cdot S(k) \cdot \frac{\ell}{k}$ for all
  $k \in \supp(S)$. Similarly, for every $k \in \supp(S)$, every
  $u \in U$ where $d(u) = k$, and every $0 \leq i < \ell-k$ we must have
  that $\sum_{j=i}^{i+d(u)-1} \mu_{h,i}(u) \geq \val$ (otherwise we obtain
  contradiction in the way indicated above), which is possible only if
  $\sum_{j=i}^{i+d(u)-1} \mu_{h,i}(u) = \val$ for all such $i$ and $u$.
  In particular, this holds for $i=0$, and the proof is finished.
\end{proof}

\noindent
Now we present a sequence of observations that reveal a certain 
form of periodicity in the structure of~$\sigma$. The next lemma
follows trivially from Lemma~\ref{lem-mu-val}.

\begin{lemma}
\label{lem-upper-tran-prob}
  For all $h \in \shistories$ and $u \in U$ we have that
  $\sigma(h)(u) \leq  \val(\sigma)$.
\end{lemma}

\begin{lemma}
\label{lem-not-earlier}
  Let $h \in \shistories$ where $\last(h) = u$. Then for every
  $hh' \in \shistories$ where $|h'| < d(u)$ we have that 
  $\last(h') \neq u$.   \end{lemma}
\begin{proof}
  Suppose the converse. Then there exist $hh' \in \shistories$ 
  and a node $u \in U$ such that $\last(h) = \last(h') = u$ and
  $|h'| < d(u)$. Due to Proposition~\ref{prop-abafy}, we have that
  $\val_{u}(\sigma_h) = \val$. Further, 
  $\sum_{i=0}^{d(u)-1} \mu_{h,i}(u) = \val$ by Lemma~\ref{lem-mu-val}.
  However, due to the existence of $h'$ we obtain that
  $\val_{u}(\sigma_h) < \sum_{i=0}^{d(u)-1} \mu_{h,i}(u)$, which is a 
  contradiction.
\end{proof}

\begin{lemma}
\label{lem-pred}
  Let $h \in \shistories$ where $\last(h) = u$.
  For all $i \geq 0$ and  $hh' \in \shistories$ where 
  $|h'| = i \cdot d(u) + d(u) - 1$ 
  we have that $\sigma(hh')(u) = \val(\sigma)$ and $u$ does
  not appear among the last $d(u)-1$ nodes of $h'$.
\end{lemma}
\begin{proof}
  By induction on $i$. In the base case ($i = 0$), we have that
  $u$ does not appear among the last $d(u)-1$ nodes of $h'$
  by Lemma~\ref{lem-not-earlier}. Further, by
  Lemma~\ref{lem-mu-val} and Lemma~\ref{lem-not-earlier} we obtain
  that $\val_{u}(\sigma_h) = \mu_{h,d(u)-1}(u)$. Hence, 
  $\mu_{h,d(u)-1}(u) = \val = \val(\sigma)$. By 
  Lemma~\ref{lem-upper-tran-prob}, this is possible only if 
  for all $hh' \in \shistories$ where $|h'| = d(u)-1$ we have
  that $\sigma(hh')(u) = \val(\sigma)$. For the inductive step,
  consider $hh'h'' \in \shistories$ where $|h'| = i \cdot d(u) + d(u) - 1$ 
  and $|h''| = d(u)$. By applying induction hypothesis to $hh'$, 
  we obtain that $\sigma(hh')(u) = \val(\sigma)$. If $u$ was revisited 
  in the last $d(u)-1$ nodes of $h''$, we would have
  $\sum_{i=0}^{d(u)-1} \mu_{hh',i} (u) > \val$, which contradicts
  Lemma~\ref{lem-mu-val}. If $\sigma(hh'h'')(u) < \val(\sigma)$,
  we obtain $\sum_{i=0}^{d(u)-1} \mu_{hh'u',i} (u) < \val$, where $u'$
  is the first node of $h''$, which again contradicts Lemma~\ref{lem-mu-val}.
\end{proof}

\begin{lemma}
\label{lem-sub-revisit}
  Let $hh' \in \shistories$ where $\last(h) = \last(h') = u$.
  Then $d(u)$ divides $|h'|$.
\end{lemma}
\begin{proof}
  Directly from Lemma~\ref{lem-pred}.
\end{proof}

\begin{lemma}
\label{lem-sub-revisit-surely}
  Let $h \in \shistories$ where $\last(h) = u$. For every $i \in \Nset$,
  there exist $hh' \in \shistories$ such that $|h'| = i \cdot d(u)$ and
  $\last(h') =u$.
\end{lemma}
\begin{proof}
  Immediate.
\end{proof}

\noindent
For the rest of this section, let us fix a history
$h = u_0 \cdots u_m \in \shistories$ such that 
every node of $U$ appears in~$h$ (such an $h$ must exist). 
For every $u \in U$, let us fix some $j \leq m$ such that $u_j = u$, and let 
$\offset(u) = j - \left\lfloor \frac{j}{d(u)} \right\rfloor \cdot d(u)$.
Note that due to Lemma~\ref{lem-sub-revisit}, the definition of $\offset(u)$
is independent of the concrete choice of~$j$. For every $k \in \supp(S)$
and every $i \in \{0,\ldots,k-1\}$, let $V_k[i]$ be the set of all
nodes $u \in U$ such that $d(u) = k$ and $\offset(u) = i$. 

\begin{lemma}
\label{lem-sub-edges}
  Let $k,k' \in \supp(S)$, $i \in \{0,\ldots, k{-}1\}$, 
  $i' \in \{0,\ldots, k'{-}1\}$, and $0 \leq \ell < k \cdot k'$ where
  $i = \ell \,\mathit{mod}\, k$ and $i' = \ell {+} 1 \,\mathit{mod}\, k'$. 
  Then for all $u \in V_k[i]$ and $u' \in V_{k'}[i']$ we have that
  $(u,u') \in E$.
\end{lemma}
\begin{proof}
  Due to Lemma~\ref{lem-sub-revisit-surely}, there exist 
  $h h' \in \shistories$ and $hh''\in \shistories$ such that 
  $|h'| = \ell$, $|h''| = \ell + 1$, $\last(h') = u$, and $\last(h'') = u'$.
  By Lemma~\ref{lem-pred}, we obtain $\sigma(hh')(u') = \val$, which
  means $(u,u') \in E$.
\end{proof}

\begin{lemma}
\label{lem-sub-size}
  For all $k \in \supp(S)$ and $i \in \{0,\ldots, k-1\}$, the set
  $V_k[i]$ contains exactly $S(k)/k$ nodes.
\end{lemma}
\begin{proof}
  By applying Lemma~\ref{lem-mu-val}.
\end{proof}

\noindent
Due to Lemma~\ref{lem-sub-size}, we have that for all $k \in \supp(S)$ and
$i \in \{0,\ldots, k-1\}$, the set $V_k[i]$ has exactly $S(k)/k$
elements, which we denote by $v_k[i,1],\ldots,v_k[i,S(k)/k]$.
Due to Lemma~\ref{lem-sub-edges}, for every pair of nodes 
$v_k[i,j]$ and $v_{k'}[i',j']$, such that $i = \ell \,\mathit{mod}\, k$ and
$i' = (\ell {+} 1) \,\mathit{mod}\, k'$ for some 
$0 \leq \ell < k\cdot k'$ we have that $(v_k[i,j],v_{k'}[i',j']) \in E$.
Hence, $(U,E)$ contains a subdigraph which is $d$-preserving isomorphic to $\ch_S$.

\section{Complexity of finding the characteristic subdigraph}
\label{app-HAM}

In this section we prove two claims leading to combined
Theorem~\ref{thm-connected} via Theorem~\ref{thm-subdigraph}.
We will focus on a subclass of patrolling problems $\game = (U,T,\su,E,d)$ 
such that $T = U$, $\supp(S) = \{k\}$.
In such a case, for a well-formed attack signature $S$,
we have that $|U|=n$ is divisible by $k$ and that the characteristic digraph
$\ch_s$ has a particularly nice description:
$\ch_s$ has a node set $u_0,\dots u_{n-1}$ and $u_iu_j$ is an arc iff
$j=(i+1)\,\mathit{mod}\,k$.

Our proofs will actually be expressed in terms of a {\em special equitable
$k$-colouring} of the complementary digraph $H=(U,\overline{E}\,)$
of the environment $E$ \,(i.e., $H$ having precisely those arcs, but not
the loops, which are absent in $E$):
Let $|U|=|V(H)|=a\cdot k$.
The task is to find a colouring $c:V(H)\to\{1,2,\dots,k\}$ of the node set
such that (a) $|c^{-1}(i)|=a$ for each $i=1,\dots,k$,
and (b) no arc $xy$ of $H$ receives colours
$c(x)=j$, $c(y)=(j\,\mathit{mod}\,k)+1$ for some $j\in\{1,\dots,k\}$
(while both $x,y$ might receive the same colour).
Comparing this with the definition of $\ch_s$ one immediately
concludes that $(U,E)$ contains a subdigraph isomorphic to $\ch_s$ if, and only if,
the complement $H$ has a special equitable $k$-colouring.

\begin{lemma}
For a simple digraph $H$ on an even number of nodes, one can find in
polynomial time a special equitable $2$-colouring of $H$, if it exists.
\end{lemma}

\begin{proof}
Note that our definition of a special $2$-colouring does not allow
for arcs having two distinct colours on their nodes, in either order.
Hence every weak component of $H$ must be monochromatic
(recall that a weak component is a connected component of the underlying
undirected graph of $H$).
The problem thus reduces to finding a subset of weak components of $H$
summing to exactly half of the nodes of $H$.
This we solve in polynomial time by two folklore algorithms;
finding the weak components by BFS, and solving the knapsack problem in
unary notation by standard dynamic programming.
\end{proof}

\begin{lemma}
Let $k \in \Nset$, $k\geq3$.
Assume a simple digraph $H$ such that $|V(H)|$ is divisible by~$k$.
Then it is $\textbf{NP}$-complete to decide whether $H$
has a special equitable $k$-colouring.
\end{lemma}

\begin{proof}
First to say, there does not seem to be an easy way how to reduce a case of
$k\geq3$ to that of $k+1$, and so we have to provide hardness reductions for
each considered value of~$k$.
We reduce from the folklore $\textbf{NP}$-complete problem of {\em
two-colouring $3$-uniform hypergraph}:
Given is a ground set $X$ and a family $\mathcal{F}$ of $3$-element subsets of $X$
(hyperedges).
The task is to decide whether the elements of $X$ can be assigned one of two
colours each such that no set in $\mathcal{F}$ is monochromatic.

\subparagraph{$(k=3)$}
For such a $3$-uniform hypergraph $(X,\mathcal{F})$ we first construct an
equivalent instance $H$ of the special equitable $3$-colouring problem.
Let $a=3|\mathcal{F}|+|X|$.\,\footnote{Although the formula for $a$ might seem arbitrary now, this
precise expression will become relevant with the case of $k=6$.}
We denote by $A_3$ the digraph of $a'=a+|\mathcal{F}|$ nodes 
$s_1,s_2,\dots,s_{a'}$ and of $a'-1$ arcs $s_1s_i$ for $i=2,\dots,a'$ ($A_3$ is a star),
and by $B$ the digraph on $a-|\mathcal{F}|$ nodes with no arcs at all.
Then we construct a digraph $G_3$ on the node set $X\cup\mathcal{F}^3$
where $\mathcal{F}^3$ is a set containing exactly three distinct copies
$f,f',f''$ of each hyperedge $f\in\mathcal{F}$.
The arcs of $G_3$ are given as follows;
for each $f=\{x_1,x_2,x_3\}\in\mathcal{F}$ there is a directed $6$-cycle
on the nodes $x_1,f,x_2,f',x_3,f''$ in this cyclic order
(a permutation of $x_1,x_2,x_3$ is irrelevant, though).
A digraph $H$ is constructed from the disjoint union of $A_3,B$ and $G_3$,
by adding arcs from the node $s_1$ to all the nodes in $X$ of $G_3$.

Then $H$ has exactly $a+|\mathcal{F}|+a-|\mathcal{F}|+|X|+3|\mathcal{F}|=3a$ nodes,
and we claim that $H$ has a special equitable $3$-colouring if, and only if,
$(X,\mathcal{F})$ is two-colourable.
In the forward direction,
up to symmetry between the colours, we may assume that $s_1$ gets colour
$1$, and so all nodes of $A_3$ have colours $1$ or $3$.
We argue the following properties:
\begin{itemize}
\item[(i)]
The nodes in $X$ can only receive colours $1,3$.
\item[(ii)]
Among the nodes of $G_3$ not in $X$, at least $|\mathcal{F}|$
of them must receive colour $2$.
\end{itemize}
Here (i) follows from the fact that each node in $X$ ends an arc starting
in $s_1$ (of $A_3$) of colour $1$.
To get (ii), notice that we have to assign colour $2$ to exactly $a$ nodes
which cannot appear in $A_3$ and in $X$ due to $s_1$ having colour $1$.
We can give colour $2$ to the nodes of $B$, yet, at least
$a-|B|=|\mathcal{F}|$ of the nodes of colour $2$ must be in $G_3\setminus X$.

Now we prove that if (i),(ii) hold true, then the hypergraph
$(X,\mathcal{F})$ is two-colourable.
Consider one of the $6$-cycles of $G_3$, say the one on the nodes
$x_1,f,x_2,f',x_3,f''$.
It cannot happen $c(f)=c(f')=2$\,---in such a case, depending on the colour $c(x_2)$, 
there would be an arc in $G_3$ coloured with a forbidden pair $1,2$ or $2,3$.
Hence each of the $6$-cycles defining the arcs of $G_3$ 
(for each $f=\{x_1,x_2,x_3\}\in\mathcal{F}$)
has at most one vertex of colour $2$, and so exactly one such.
Up to symmetry, let $c(f)=2$ in (any) one of the cycles.
Then $c(x_1)\not=c(x_2)$, since otherwise $c(x_1)=c(x_2)\in\{1,3\}$ would
again give a forbidden pair of colours $1,2$ or $2,3$, respectively.
Consequently, taking the colouring $c$ restricted to $X$, no hyperedge in
$\mathcal{F}$ is monochromatic and $(X,\mathcal{F})$ is two-colourable.

Conversely, consider a two-colourable $3$-uniform hypergraph $(X,\mathcal{F})$.
Let the colours occuring in $X$ be $1$ and $3$.
We extend this to a special equitable $3$-colouring of our digraph $H$ as
follows.
If a hyperedge $f=\{x_1,x_2,x_3\}\in\mathcal{F}$ is coloured $1,1,3$,
then we assign colours $1,1,1,3,3,2$ in order to the $6$-cycle on the nodes
$x_1,f,x_2,f',x_3,f''$ in $G_3$.
If this $f=\{x_1,x_2,x_3\}\in\mathcal{F}$ is coloured $1,3,3$,
then we assign colours $1,1,3,3,3,2$ in order to the same $6$-cycle.
We finally assign colour $1$ to $s_1$, colour $2$ to all nodes of $B$
(and so $c^{-1}(2)=a$),
and an arbitrary choice of colours $1,3$ to the remaining nodes of $A_3$
in order to ``balance'' $c^{-1}(1)=c^{-1}(3)=a$.

\subparagraph{$(k=4)$}
Second, we modify the previous construction of $H$ for the case of $k=4$.
We use the same $B$ and $G_3$.
We replace $A_3$ with a digraph $A_4$ which is the complete digraph on 
$a'=a+|\mathcal{F}|$ nodes $s_1,s_2,\dots,s_{a'}$, too.
Then we add a new digraph $C_4$ formed by $a$ nodes with no arcs between.
$H$ is constructed from a disjoint union of $A_4,C_4,B$ and $G_3$ by
adding all the arcs from $s_1$ to $C_4$ and all the arcs between the nodes
of $A_4$ and of $X\subseteq V(G_3)$ in both directions.
Clearly, $H$ has $4a$ nodes.

Consider a special equitable $4$-colouring of $H$.
Again, up to symmetry, let the colour of $s_1$ be $1$.
Then whole $A_4$ and $X$ may only receive colours $1$ or $3$ and (i) holds
true again.
Since no node of $C_4$ may be coloured $2$ due to the existence of an arc
from $s_1$, and since $B$ (which may be coloured by $2$) has size $a-|\mathcal{F}|$,
we get (ii), too.
Now, notice that the argument following (i),(ii) above did not use the pair
of colours $3,1$ as forbidden, and so it applies now as well;
$(X,\mathcal{F})$ is two-colourable.

Conversely, consider a two-colourable $3$-uniform hypergraph $(X,\mathcal{F})$.
Then, exactly as in the case of $k=3$, we get a valid colouring of $A_4\cup
B\cup G_3$ which we complement by assigning colour $4$ to whole $C_4$.
This results in a special equitable $4$-colouring of $H$.

\subparagraph{$(k\geq5)$}
Third, we define a general construction for all the values $k=5,6,\dots$.
We use the same gadgets $G_3$, $B$, and $A_4$, and introduce $k-3$ disjoint
copies of $C_4$ which we denote by $C_4,C_5$ and $D_5,\dots,D_{k-1}$.
Again, on the disjoint union of all these digraphs (which has $k\cdot a$
nodes) we define $H$ by adding
\begin{itemize}
\item all the arcs between the nodes
of $A_4$ and of $D_5\cup\dots\cup D_{k-1}$ in both directions,
\item all the arcs between the nodes
of $A_4\cup D_5\cup\dots\cup D_{k-1}$ and the nodes $X$ of $G_3$ in both directions,
\item all the arcs between the nodes
of $D_5\cup\dots\cup D_{k-1}$ and of $B$ in both directions,
\item all the arcs between $s_1$ and the nodes of $C_4$ in both directions,
and the same between $s_2$ and $C_5$.
\end{itemize}

Consider a special equitable $k$-colouring of $H$.
For simplicity we call the forbidden pairs of colours \mbox{$j$,
$(j\,\mathit{mod}\,k)+1$} as {\em adjacent}.
Since $A_4\cup D_5\cup\dots\cup D_{k-1}$ has $(k-4)a+1$ nodes, at least $k-3$
distinct colours must occur there.
However, $A_4$ (of $>\!a$ nodes) itself gets at least two distinct
non-adjacent colours $c_1,c_2$ which cannot be adjacent to any of the colours 
occuring in $D_5\cup\dots\cup D_{k-1}$ other than $c_1,c_2$.
A simple case analysis shows that the only valid choice of colours is
$c_1=1$, $c_2=3$ and remaining $5,6,\dots,k-1$, up to rotation symmetry.
Consequently, $A_4$ holds only colours $1,3$ and each of the colours
$5,6,\dots,k-1$ occurs somewhere in $D_5\cup\dots\cup D_{k-1}$.
In particular, no node of $A_4\cup D_5\cup\dots\cup D_{k-1}$ is coloured~$2$.

Which nodes could have colour $2$?
Due to the arcs to and from $s_1,s_2$ in $A_4$, all the nodes of colour $2$
belong to $B\cup(G_3\setminus X)$, and since $G_3\setminus X$ has
$3|\mathcal{F}|<a$ nodes, we have $c^{-1}(2)\cap B\not=\emptyset$.
This has a twofold consequence; first, (ii) holds true also in this case,
and second, colours $1,3$ cannot occur in $D_5\cup\dots\cup D_{k-1}$.
Then, by simple counting, $c^{-1}(5)\cup\dots\cup c^{-1}(k-1)$ must be exactly
the node set of $D_5\cup\dots\cup D_{k-1}$,
and hence $X$ cannot get any of the colours $5,\dots,k-1$.
Neither colours $2,4$ or $k$ could occur in $X$ due to the arcs to and from $A_4$,
which concludes that (i) holds true, too.
Theorefore, $(X,\mathcal{F})$ is two-colourable.

Conversely, consider a two-colourable $3$-uniform hypergraph $(X,\mathcal{F})$.
We colour $G_3\cup A_4$ by $1,2,3$ as above while giving $c(s_1)=1$ and $c(s_2)=3$.
Then we assign colour $2$ to whole $B$, colour $4$ to whole $C_4$, colour
$k$ to whole $C_5$, and colours $j$ to whole $D_j$ for $j=5,\dots,k-1$.
Again, this results in a special equitable $k$-colouring of $H$.
\end{proof}

\end{document}